\documentclass[12pt,fullpage]{article}
\usepackage{float}
\usepackage{caption}
\usepackage{booktabs}
\usepackage{ltablex}
\usepackage{adjustbox}

\usepackage{microtype}
\usepackage{fullpage}
\usepackage{setspace}
\usepackage{natbib}
\usepackage{hyperref}
\usepackage{framed}
\usepackage{bm}
\usepackage{blkarray}% http://ctan.org/pkg/blkarray
\newcommand{\matindex}[1]{\mbox{\scriptsize#1}}% Matrix index
\usepackage{arydshln}
\usepackage{mleftright}
\usepackage{kbordermatrix}
\renewcommand{\kbldelim}{(}% Left delimiter
\renewcommand{\kbrdelim}{)}% Right delimiter

\usepackage{array,xspace,hhline,tikz,colortbl,tabularx,booktabs,fixltx2e,amsmath,amssymb,amsfonts,amsthm}
\usepackage{verbatim,ifthen}
\usepackage{pifont}
\usepackage{ifthen}
\usepackage{pgflibraryshapes}

\usepackage{nicefrac}
\usetikzlibrary{arrows}
\usepackage{ltxtable}
\usepackage{longtable}

	\usepackage{varioref}

\definecolor{light-gray}{gray}{0.9}

\newtheorem{definition}{Definition}%

	\newtheorem{lemma}{Lemma}%
	\newtheorem{proposition}{Proposition}%
	\newtheorem{example}{Example}
	\newtheorem{remark}{Remark}

	\newcommand{\midd}{\mathbin{:}}

\newcommand{\sd}[0]{\ensuremath{\mathit{sd}}}

\title{\textbf{Random Matching under Priorities:\\Stability and No Envy Concepts}\thanks{This work has been partly supported by COST Action IC1205 on Computational Social Choice. The authors thank Panos Protopapas, Madhav Raghavan, Jan Christoph Schlegel, and Zhaohong Sun for their feedback and comments. We particularly thank Battal Do\u{g}an for detailed comments and for providing an example that helped establish one of our results (Proposition~\ref{prop:weakfractional-notto-weakexpost}).}}

\author{Haris Aziz\thanks{\textit{e-mail}: \href{mailto:haris.aziz@data61.csiro.au}{haris.aziz@data61.csiro.au}; Data61, CSIRO, and UNSW, Computer Science and Engineering, Building K17, Sydney NSW 2052, Australia. Haris Aziz is supported by a Julius Career Award.} \and Bettina Klaus\thanks{\emph{Corresponding author}, \textit{e-mail}: \href{mailto:Bettina.Klaus@unil.ch}{Bettina.Klaus@unil.ch}; Faculty of Business and Economics, University of Lausanne, Internef, CH-1015 Lausanne, Switzerland. Bettina Klaus gratefully acknowledges financial support from the Swiss National Science Foundation (SNFS). Furthermore, Bettina Klaus  gratefully acknowledge the hospitality of Stanford University where part of this paper was written.}}

\date{}
\begin{document}

\sloppy

\maketitle

\begin{abstract}
\noindent We consider stability concepts for random matchings where agents have preferences over objects and objects have priorities for the agents.
When matchings are deterministic, the standard stability concept also captures the fairness property of no (justified) envy. When matchings can be random, there are a number of natural stability / fairness concepts that coincide with stability / no envy whenever matchings are deterministic.
We formalize known stability concepts for random matchings for a general setting that allows weak preferences and weak priorities, unacceptability, and an unequal number of agents and objects. We then present a clear taxonomy of the stability concepts and identify logical relations between them.
Furthermore, we provide no envy / claims interpretations for some of the stability concepts that are based on a consumption process interpretation of random matchings. Finally, we present a transformation from the most general setting to the most restricted setting, and show how almost all our stability concepts are preserved by that transformation.\smallskip

\noindent \textit{JEL} Classification Numbers: C63, C70, C71, and C78.\smallskip

\bigskip
\end{abstract}

\section{Introduction}

We consider a model of matching agents to objects in which agents have preferences over objects and objects have priorities for the agents. This general model has many applications, e.g., for school choice \citep{AbSo03b}: in Boston \citep*{APR05b} and New York \citep*{APR05a}, centralized matching schemes are employed to assign students to schools on the basis of students' preferences over schools and students' priorities to be admitted to any given school.  A school's priority for a student might include issues such as geographical proximity and whether the student has a sibling at the school already, among others. Two surveys on school choice can be found in \citet{Abd13} and \citet{Pat11}.\vspace{-0.2cm}

\paragraph{Stability and no envy:} For the most basic model we discuss in Section~\ref{section:model}, the fundamental stability concern is the following: no agent $i$ should prefer an object $o$ matched to another agent $j$ who has lower priority for the object than $i$. In school choice, this notion of stability can be interpreted as the elimination of justified envy \citep{AbdulkadirogluSonmezAER2003,BS99}: a student can justifiably envy the match of another student to a school if he likes that school better than his own match and he has a higher priority (with a lower priority, envy might be present as well but is not justifiable). For the most general model we discuss in Section~\ref{subsection:weakdifferentnumbers}, (weak) stability is equivalent to individual rationality, non-wastefulness, and no justified envy. To simplify language, we from now on will refer to no justified envy simply as no envy. While the important role of stability in matching problems has long been recognized,\footnote{See, e.g., the advanced information document for the Sveriges Riksbank Prize in Economic Sciences in Memory of Alfred Nobel 2012 awarded to A.E.\ Roth and L.S.\ Shapley ``for the theory of stable allocations and the practice of market design'' \href{http://www.nobelprize.org/nobel_prizes/economic-sciences/laureates/2012/advanced-economicsciences2012.pdf}{URL: Advanced Information ``Nobel Prize 2012''}.} no envy, which is a relaxation of stability, has only recently gained independent interest. In particular, no envy has been considered in constrained matching models \citep*{EhlersHalafirYenmezYildirimJET2014,KamadaKojimaJET2016} and in senior level labor markets \citep*{BlumRothRothblumJET1997} and shown to have similar structural properties as stability \citep{WuRothWP2017}. The well-known deferred-acceptance algorithm \citep{GS62} computes a deterministic matching that is (weakly) stable and hence envy free.\vspace{-0.2cm}

\paragraph{Random and fractional matchings:} Most articles on school choice and similar models have considered deterministic matchings.
Instead, we consider random matchings that specify the probability of each agent being matched to the various objects.
Random matchings are useful to consider for several reasons.
Firstly, randomization allows for a much richer space of possible outcomes and may be essential to achieve fairness properties such as anonymity\footnote{An \emph{anonymous} mechanism does not depend on the names of the agents.} and (ex-ante) equal-treatment-of-equals.\footnote{A mechanism satisfies \emph{(ex-ante) equal-treatment-of-equals} if two equal agents (with the same preferences) receive the same (ex-ante) allocation.} It thus allows for a richer set of mechanisms with the possibility of better properties: as pointed out by \citet{KeUn15a}, a broader view of fairness has largely been ignored in prior work.
Secondly, the framework of random matchings also helps to reason about fractional matchings that capture time sharing arrangements \citep*{RothetalMOR93,TeSe98a,DoYi16a}. For example, an agent may allocate his time among several of his matches rather than exclusively being matched to a single object. Mathematically, we then simply consider the probability of an agent getting an object as the match of a corresponding fraction of time to the object.
Thirdly, randomization has proved to be useful to circumvent impossibility results in social choice~\citep{BoMo01a,Bran13a,DNS06a,Gibb77a}.

Whereas particular stability concepts for random / fractional matchings have been introduced and studied in various papers, the picture of how exactly they relate to each other and how their formulations change for various models (allowing for indifferences, unacceptability, and a different number of agents / objects) has, to the best of our knowledge, not been studied until now. This gap in the literature is especially important to address with the renewed interest in recent years in random matching mechanisms.\vspace{-0.2cm}

\paragraph{Overview of the article:}  We consider some existing stability concepts (\textit{ex-post and fractional stability}, \citeauthor{RothetalMOR93}, \citeyear{RothetalMOR93}, and \citeauthor{TeSe98a}, \citeyear{TeSe98a}; \textit{ex-ante / strong stability}, \citeauthor{RothetalMOR93}, \citeyear{RothetalMOR93}, and \citeauthor{KeUn15a}, \citeyear{KeUn15a}; and \textit{claimwise stability}, \citeauthor{Afac15a}, \citeyear{Afac15a}) and also propose a new one, \textit{robust ex-post stability},  that is nested between ex-ante stability and ex-post stability. Many of the concepts have been defined and then subsequently studied only for restricted settings that use one, some, or all of the following restrictions: (1) preferences are strict, (2) priorities are strict, (3) there is an equal number of agents and objects, (4) all objects and agents are acceptable to each other. We generalize all the stability concepts mentioned above to the general random matching setting that allows for indifferences in preferences and priorities, and allows for unacceptability as well as for an unequal number of agents and objects. The general setting includes as a special case the \emph{hospital-resident setting} in which hospitals have multiple positions but residents are indifferent among all such positions at the same hospital; another example is the previously mentioned \emph{school choice setting}. The general model and our insights into the corresponding stability notions will provide a crucial stepping stone for further work on axiomatic, algorithmic, and market design aspects of random stable matching.

In particular, we present a taxonomy of the stability concepts for random matching of objects when objects have priorities for agents. Our study helps clarify the relations between the different stability / fairness concepts mentioned above. This taxonomy also points the market designer to consider a scale of criteria of different ``stability-strengths'' while designing mechanisms, which additionally could also satisfy other properties (e.g., efficiency or strategic robustness).
Furthermore, we provide no envy / claims interpretations for some of the stability concepts that are based on a consumption process interpretation of random matching. Finally, we present a transformation from the most general setting (without any restrictions) to the most restricted setting (with restrictions (1) -- (4)), and show how almost all our stability concepts are preserved by that transformation.\vspace{-0.2cm}

\paragraph{The article proceeds as follows.} In Section~\ref{section:model} we introduce the base model in which: (1) preferences are strict, (2) priorities are strict, (3) there is an equal number of agents and objects, (4) all objects and agents are acceptable to each other. For this model, stability and no envy coincide. We introduce the random stability concepts ex-ante stability, robust ex-post stability, ex-post stability, fractional stability, and claimwise stability and provide  consumption process  interpretations for fractional and claimwise stability that are based on specific envy / claim notions. We discuss the convexity of the stability concepts for random matchings and present a complete taxonomy of the stability concepts for our base model (see Figure~\ref{fig:part-relations}).

Then, we extend the base model in two ways. First, in Section~\ref{subsection:weakprefpri}, we drop model assumptions (1) and (2) and allow preferences and priorities to be weak. The switch from strict preferences / priorities to weak ones requires various adjustments in definitions and the consumption process interpretation, but once these adjustments are made, results change very little (see Figure~\ref{fig:weakpart-relations}). Second, we drop model assumptions (3) and (4) and [allow for an unequal number of agents and objects] and [that agents / objects find some objects / agents unacceptable]. With this change, we add the well-known criteria of non-wastefulness and individual rationality: for the general model that is considered now, (weak) stability is equivalent with no envy, non-wastefulness, and individual rationality. We then formalize all stability concepts with the appropriate additional requirements of non-wastefulness and/or individual rationality when necessary to preserve the hierarchy we established in the base model. We use a transformation from the most general setting to the most restricted setting for random matchings to show how almost all our stability concepts are preserved by that transformation\footnote{With one exception: due to some lack of symmetry in the definition of claimwise (weak) stability, full equivalence under the transformation fails (Proposition~\ref{Prop:weakclaimwisestableiffclaimwisestable} and Example~\ref{examle:weakclaimiseone-sided}).} and to establish a complete taxonomy of stability concepts for the general model (see Figure~\ref{fig:weakassociated-random}).

\section{The base model}\label{section:model}

Let $N$ be a set of $n$ \textbf{\textit{agents}} and $O$ be a set of $n$ \textbf{\textit{objects}}. Each agent $i\in N$ has \textbf{\textit{preferences}} $\succ_i$ over $O$ and each object $o\in O$ has \textbf{\textit{priorities}} $\succ_o$ for $N$ (although using the same notation, the reason we use the term priorities instead of the term preferences is that objects are not considered as economic agents in our model). Agents' preferences are strict orders over $O$ and objects' priorities are strict orders over $N$.\medskip

We'll explain later on how the model and results extend when preferences and priorities can be weak (Section~\ref{subsection:weakprefpri}) and when allowing for unacceptability as well as different numbers of agents and objects (Section~\ref{subsection:weakdifferentnumbers}).\medskip

A \textbf{\textit{random matching}} $p$ is a bistochastic $n\times n$ matrix $[p(i,o)]_{i\in N,o\in O}$, i.e.,
\begin{equation}\label{LE1}\mbox{for each pair }(i,o)\in N\times O,\ p(i,o)\geq 0,\end{equation}
\begin{equation}\label{LE2}\mbox{for each }i\in N,\ \sum_{o\in O}p(i,o)=1,\mbox{ and }\end{equation}
\begin{equation}\label{LE3}\mbox{for each }o\in O,\ \sum_{i\in N}p(i,o)=1.\end{equation}
Random matchings are often also referred to as \textit{fractional matchings} \citep{RothetalMOR93,TeSe98a}.
For each pair $(i,o)\in N\times O$, the value $p(i,o)$ represents the probability of object $o$ being matched to agent $i$ and agent $i$'s \textbf{\textit{match}} is the probability vector $p(i)=(p(i,o))_{o\in O}$.
A random matching $p$ is \textbf{\textit{deterministic}} if for each pair $(i,o)\in N\times O$, $ p(i,o) \in \{0,1\}$. Alternatively, a deterministic matching is an integer solution to linear inequalities (\ref{LE1}), (\ref{LE2}), and (\ref{LE3}).\medskip

By \cite{Birkhoff1946} and \cite{Neumann1953}, each random matching can be represented as a convex combination of deterministic matchings: a \textbf{\emph{decomposition}} of a random matching $p$ into deterministic matchings $P_j$ ($j\in \{1,\ldots,k\}$) equals a sum $p=\sum_{j=1}^k \lambda_jP_j$ such that for each $j\in \{1,\ldots,k\}$, $\lambda_j\in (0,1]$ and $\sum_{j=1}^k\lambda_j=1$.

\subsection{Stability concepts}\label{section:stability}

\begin{definition}[\textbf{No envy / stability for deterministic matchings}]
\normalfont A deterministic matching $p$ has \emph{\textbf{no envy}} or is \textbf{\emph{stable}} if there exists no agent $i$ who is matched to object $o'$ but prefers object $o$ while object $o$ is matched to some agent $j$ with lower priority than $i$, i.e., there exist no $i,j\in N$ and no $o,o'\in O$ such that $p(i,o')=1$, $p(j,o)=1$, $o\succ_i o'$, and $i\succ_{o}j$.\label{def:stability}
\end{definition}

\emph{Stability} was first introduced for two-sided matching markets by \citet{GS62}. The terminology of \emph{no justified envy} is usually used in the context of the so-called school choice model \citep{BS99,AbdulkadirogluSonmezAER2003}. Note that we use the shorter expression no envy for the somewhat more precise no justified envy \citep[see also][]{WuRothWP2017}.\medskip

A deterministic matching $p$ is stable if and only if it satisfies the following inequalities \citep{RothetalMOR93}:\footnote{\citet{RothetalMOR93} consider a more general model that lies between the models we discuss in Section~\ref{subsection:weakprefpri} and Section~\ref{subsection:weakdifferentnumbers}. We here use the restriction of their original inequalities to our base model.} for each pair $(i,o)\in N\times O$,
\begin{equation}\label{LE4}
p(i,o)+\sum_{o':o'\succ_i o}p(i,o')+ \sum_{j:j\succ_{o}i}p(j,o)\geq 1.
\end{equation}

The well-known deferred-acceptance algorithm \citep{GS62} computes a deterministic matching that is stable.\medskip

We now define five stability concepts for random matchings that all coincide with deterministic stability when the random matching is deterministic.\medskip

The first stability concept for random matchings we consider was discussed by \citet{RothetalMOR93} under the name of \textit{strong stability} for the marriage market matching model. Recently, for a school choice model, \citet{KeUn15a} obtained the same stability concept by extending no envy from matched whole objects to matched probability shares of objects; the intuition here is that a higher priority agent $i$ envies a lower priority agent $j$ for any probability share of object $o$ that agent $j$ has if he would like to get a higher probability of it himself. We will discuss and prove later in this section that \citet{AhFl03a} introduced a stability concept for a more general model of so-called hypergraphic preference systems that coincides with ex-ante stability.

\begin{definition}[\textbf{No ex-ante envy / ex-ante stability}]
\normalfont	A random matching $p$ has \emph{\textbf{no ex-ante envy}} or is \emph{\textbf{ex-ante stable}} if there exists no agent $i$ who is matched with positive probability to object $o'$ but prefers object $o$ to object $o'$ while object $o$ is matched with positive probability to some agent $j$ with lower priority than $i$, i.e., there exist no $i,j\in N$ and no $o,o'\in O$ such that $p(i,o')>0$, $p(j,o)>0$, $o\succ_i o'$, and $i\succ_oj$.\label{def:ex-antestability}
\end{definition}

Although normatively appealing, the notion of ex-ante stability is demanding. It follows from \citet[][Corollary~21]{RothetalMOR93} that each agent can receive probability shares of, at most, two objects and vice versa, each object is assigned with positive probability to at most two agents. In other words, an ex-ante stable random matching is almost deterministic. \citet{SchlegelLotteries16} generalizes this result to the more general set-up with quotas and priority ties as follows: ex-ante stable random matchings have small support, meaning that only few agent-object pairs have a positive probability of being matched. The number of pairs in the support depends on how many indifferences in the priorities the random matching exploits. In the extreme case where no object is matched with positive probability to two equal priority agents, the probability distribution is almost deterministic. Otherwise, the size of the support is completely determined by the size of the lowest priority classes at which agents are matched to the respective objects. This result can be interpreted as an impossibility result: with ex-ante stability one cannot go much beyond randomly breaking ties and implementing a (deterministically) stable matching with respect to the broken ties.\medskip
	
The second stability concept for random matchings we consider is \emph{ex-post stability}.

\begin{definition}[\textbf{Ex-post stability}]
\normalfont	A random matching $p$ is \emph{\textbf{ex-post stable}} if it can be decomposed into deterministic stable matchings.\label{def:expoststability}
\end{definition}

For a one-to-one marriage market setup, \cite{DoYi16a} show that for each ex-post stable random matching, there is a utility profile consistent with the ordinal preferences such that no group of agents consisting of equal numbers of men and women can deviate to a random matching among themselves and make each member better off in an expected utility sense. For real world applications, ex-post stability has the desirable feature that a deterministic stable matching can be drawn from the existing probability distribution and be implemented. Various authors \citep{Vand89a,Rothblum1992,RothetalMOR93} proved  that ex-post stability is in fact characterized by inequalities (\ref{LE1}), (\ref{LE2}), (\ref{LE3}), and (\ref{LE4}): the extreme points of the polytope defined by these linear inequalities are exactly the (incidence vectors of the) deterministically stable matchings. A random matching is hence ex-post stable if it is a (not necessarily integer) solution to the linear inequalities. However, the fact that an ex-post stable matching is also a solution to a system of inequalities and vice versa is not a trivial result and one can use the inequalities (\ref{LE1}), (\ref{LE2}), (\ref{LE3}), and (\ref{LE4}) to define a separate stability concept. This leads to our third stability concept for random matchings, \emph{fractional stability}.

\begin{definition}[\textbf{Fractional stability and violations of fractional stability}]
\normalfont	A random matching $p$ is \textit{\textbf{fractionally stable}} if for each pair $(i,o)\in N\times O$,
\begin{equation}\tag{\ref{LE4}}
p(i,o)+\sum_{o':o'\succ_i o}p(i,o')+ \sum_{j:j\succ_{o}i}p(j,o)\geq 1,\footnote{If inequalities (\ref{LE4}) hold, \citet{RothetalMOR93} have also shown that for each pair $(i,o)\in N\times O$, $p(i,o)>0$ implies $p(i,o)+\sum_{o':o'\succ_i o}p(i,o')+ \sum_{j:j\succ_{o}i}p(j,o)= 1$.}
\end{equation}or more compactly,
\begin{equation}\label{LE5}
\sum_{o':o'\succ_i o}p(i,o') \geq  \sum_{j:j\prec_o i}p(j,o).
\end{equation}
A \textit{\textbf{violation of fractional stability} occurs if there exists a} pair $(i,o)\in N\times O$ such that
\begin{equation}\label{LE6}
\sum_{j:j\prec_o i}p(j,o)>\sum_{o
:o'\succ_i o}p(i,o').
\end{equation}\label{def:fractionalstability}
\end{definition}

We next explain fractional stability as a no envy notion. To this end, we first need to explain what we mean by the term \textit{consumption process}.

\paragraph{Consumption process:} an agent $i$'s match $p(i)$ can be obtained by the following (stepwise) \textit{consumption process}. We imagine that each object is represented by one unit of a homogeneously divisible pie that agents can consume and each agent wants exactly one unit of pie in total. The probability shares agents receive at $p$ are the fractions of the pies that the agents receive when they eat from one pie at a time, at equal speed, and in decreasing preference order, i.e., we imagine the match of each agent is the result of a consumption process at which they first consume the best object with positive probability share at $p$, then the second best object with positive probability share at $p$, etc. For example, consider an agent $i$ with $\succ_i: x\ y\ z$ and $p(i,x)=\frac{1}{8}$, $p(i,y)=\frac{1}{2}$, and $p(i,z)=\frac{3}{8}$. Then, in the consumption process, agent $i$ consumes first $\frac{1}{8}$ of object pie $x$, second $\frac{1}{2}$ of object pie $y$, and third $\frac{3}{8}$ of object pie $z$.\bigskip

Inequality (\ref{LE6}) implies $\sum_{o':o'\succsim_i o}p(i,o')<1$, i.e., agent $i$ receives some fraction of an object in his strict lower contour set at $o$ (if not, this would imply that $\sum_{o':o'\succ_i o}p(i,o')+p(i,o)=1$ and hence, $\sum_{j:j\prec_o i}p(j,o)+p(i,o)>1$; a contradiction to feasibility). Thus, agent $i$ would want to consume more of object $o$. Inequality (\ref{LE6}) also implies $\sum_{j:j\succsim_o i}p(j,o)<1$, i.e., object $o$ receives some fraction of an agent in its strict lower contour set at $i$. Thus, object $o$ would want to consume more of agent $i$.  Moreover, strict inequality (\ref{LE6}) encodes the following envy notion: using consumption process language, as long as agent $i$ consumes objects that are better than $o$ he does not envy the set of lower priority agents to jointly consume fractions of $o$, however, once the lower priority agents have consumed as much of $o$ as agent $i$'s strict upper contour set at $o$, agent $i$ starts having envy towards them for any additional amounts of $o$ (unless agent $i$ can fill his remaining probability quota with object $o$).\medskip

\begin{remark}[\textbf{A symmetric reformulation of fractional stability and its violations}]\label{remark:symmetricFractional}
\normalfont In the definition of fractional stability by inequalities (\ref{LE5}) and of a violation of fractional stability by inequality (\ref{LE6}) we have taken the viewpoint of an agent who considers the consumptions of lower priority agents for an object. The symmetric formulations when taking the viewpoint of an object that ``considers'' the matches of lower preferred objects to an agent are as follows. A random matching $p$ is fractionally stable if for each pair $(i,o)\in N\times O$,
\begin{equation}\tag{\ref{LE5}'}
\sum_{j:j\succ_o i}p(j,o) \geq  \sum_{o':o'\prec_i o}p(i,o').
\end{equation}
We can write a violation of fractional stability as, there exists a pair $(i,o)\in N\times O$ such that
\begin{equation}\tag{\ref{LE6}'}
\sum_{o':o'\prec_i o}p(i,o')>\sum_{j:j\succ_o i}p(j,o).
\end{equation}\hfill~$\diamond$
\end{remark}

\citet{AhFl03a} introduced a stability concept that they also called \textit{fractional stability} for a more general model of so-called hypergraphic preference systems. \citet{BiFl2016} extended the Aharoni-Fleiner notion of fractional stability to an even more general model of NTU coalition formation games. We note here that fractional stability as defined by \citet{AhFl03a} and \citet{BiFl2016} is \emph{not} equivalent to fractional stability considered in this paper. In fact, it is equivalent to ex-ante stability.

\begin{definition}[\textbf{Aharoni-Fleiner fractional stability}]\normalfont
A random matching $p$ is \textbf{Aharoni-Fleiner fractionally stable} if for each pair $(i,o)\in N\times O$,  $$\sum_{o':o' \succsim_i o}p(i,o')=1\mbox{ or } \sum_{j:j\succsim_o i}p(j,o)=1.$$\label{def:AFfractionalstability}
	\end{definition}

\begin{proposition}
A random matching is Aharoni-Fleiner fractionally stable if and only if it has no ex-ante envy.
\end{proposition}

\begin{proof}[\textbf{Proof}]Suppose random matching $p$ has ex-ante envy. Then, there exist $i,j\in N$ and $o,o'\in O$ such that $p(i,o')>0$, $p(j,o)>0$, $o\succ_i o'$, and $i\succ_o j$. Thus, there exists a pair $(i,o)\in N\times O$ such that $\sum_{o'':o'' \succsim_i o}p(i,o'')<1$ and $\sum_{k:k\succsim_o i}p(k,o)<1$. Hence, $p$ is not Aharoni-Fleiner fractionally stable.

Suppose random matching $p$ is not Aharoni-Fleiner fractionally stable. Then, there exists a pair $(i,o)\in N\times O$ such that $\sum_{o'':o''\succsim_i o}p(i,o'')<1$ and $\sum_{k:k\succsim_o i}p(k,o)<1$. Thus, there exist $i,j\in N$ and $o,o'\in O$ such that $p(i,o')>0$, $p(j,o)>0$, $o\succ_i o'$, and $i\succ_o j$. Hence, $p$ has ex-ante envy.
\end{proof}

Our fourth stability concept for random matchings is based on a new stability concept suggested by \cite{Afac15a} for a model with weak priorities. Here, we focus exclusively on the strict priority part of his stability concept even though we use the same name (Afacan has some additional conditions addressing equal priority agents that capture aspects of ``equal treatment of equals'' that are not related to stability). According to \cite{Afac15a}, an agent $i\in N$ has a \emph{claim} against an agent $j\in N$, if there exists an object $o\in O$ such that $i\succ_o j$ and
\begin{equation}\label{LE7}
p(j,o)>\sum_{o':o'\succ_i o}p(i,o').
\end{equation}
A random matching is \emph{claimwise stable} if it does not admit any claim.

\begin{definition}[\textbf{Claimwise stability}]
\normalfont	A random matching $p$ is \textit{\textbf{claimwise stable}} if for each pair $(i,o)\in N\times O$  and each $j\in N$ such that $i\succ_oj$,
\begin{equation}\label{LE8}
\sum_{o':o'\succ_i o}p(i,o') \geq  p(j,o).
\end{equation}\label{def:claimwisestable}
\end{definition}

Claimwise stability is a no envy notion based on having a claim. Inequality (\ref{LE7}) implies $\sum_{o':o'\succsim_i o}p(i,o')<1$, i.e.,  agent $i$ receives some fraction of an object in his strict lower contour set at $o$ (if not, this would imply that $\sum_{o':o'\succ_i o}p(i,o')+p(i,o)=1$ and hence, $p(j,o)+p(i,o)>1$; a contradiction to feasibility).
Thus, agent $i$ would want to consume more of object $o$.
Moreover, strict inequality (\ref{LE7}) encodes the following envy notion: using consumption process language, as long as agent $i$ consumes objects that are better than $o$ he does not envy lower priority agent $j$ to consume fractions of $o$, however, once agent $j$ has consumed as much of $o$ as agent $i$'s strict upper contour set at $o$, agent $i$ starts having envy towards him for any additional amounts of $o$ (unless agent $i$ can fill his remaining probability quota with object $o$).\medskip

Our fifth stability concept for random matchings, \emph{robust ex-post stability}, is a natural strengthening of ex-post stability.\footnote{\citet{KeUn15a} pointed out that ``Although ex post stability is a meaningful interpretation of fairness for deterministic outcomes, for lottery mechanisms such as those used for school choice, its suitability as the right fairness notion is less clear.'' They then proceed to analyze the stronger stability concept of ex-ante stability, which is a very strong stability requirement. We show that robust ex-post stability is weaker than ex-ante stability and stronger than ex-post stability and hence it is a good compromise between these competing stability concepts. Robust ex-post stability strengthens ex-post stability in a similar way as robust ex-post efficiency \citep{AMXY15a} strengthens ex-post efficiency.}

\begin{definition}[\textbf{Robust ex-post stability}]
\normalfont A random matching $p$ is \emph{\textbf{robust ex-post stable}} if all its decompositions are into deterministic and stable matchings.\label{def:robustexpoststability}
\end{definition}
	
It follows easily that if we restrict attention to deterministic matchings, then all the stability concepts for random matchings coincide with stability / no envy (Definition~\ref{def:stability}). For completeness, we provide a short proof of Proposition~\ref{proposition1} at the end of Section~\ref{section:taxonomystability}.

\begin{proposition}\label{proposition1}
For deterministic matchings, all the stability concepts for random matchings coincide with stability / no envy for deterministic matchings.
\end{proposition}

Next, we say that a stability concept $*$ is \textit{convex} if the convex combination of $*$-stable matchings is $*$-stable as well ($*$-stability stands for any of our stability concepts for random matchings). Since the stability constraints for fractional and claimwise stability are linear, there are simple (linear) arguments why both stability concepts are convex. We will later show that ex-ante stability and robust ex-post stability are not convex.

\begin{lemma}\label{lemma:frac-convex}
Fractional stability is convex.
\end{lemma}
\begin{proof}[\textbf{Proof}]
Let $p$ and $q$ be fractionally stable random matchings. Then, by Inequality (\ref{LE5}), for each pair $(i,o)\in N\times O$,
$$\sum_{o':o'\succ_i o}p(i,o') \geq  \sum_{j:j\prec_o i}p(j,o)$$ and
$$\sum_{o':o'\succ_i o}q(i,o') \geq  \sum_{j:j\prec_o i}q(j,o).$$				Then it follows that for each $\lambda\in [0,1]$,
$$\lambda\big[\sum_{o':o'\succ_i o}p(i,o')\big]+(1-\lambda)\big[\sum_{o':o'\succ_i o}q(i,o')\big] \geq  \lambda\big[\sum_{j:j\prec_o i}p(j,o)\big] + (1-\lambda)\big[\sum_{j:j\prec_o i}q(j,o)\big]$$
and
$$\sum_{o':o'\succ_i o}\big[\lambda p + (1-\lambda) q\big](i,o') \geq \sum_{j:j\prec_o i}\big[\lambda p+(1-\lambda)q\big](j,o).$$
\end{proof}

The proof for the convexity of claimwise stability is similar and we omit it.

\begin{lemma}\label{lemma:claimwise-convex}
Claimwise stability is convex.
\end{lemma}

\begin{remark}[\textbf{Core and stability concepts based on stochastic dominance and vNM preferences / priorities}]\label{remark:core}\normalfont A deterministic matching is in the \textit{\textbf{core}} if no coalition of agents and objects can improve by rematching among themselves, i.e., a deterministic matching $p$ is in the core if there exists no set $N'\cup O' \subseteq N\cup O$ and no deterministic matching $p' \neq p$ such that (i) for each $i' \in N'$, $\sum_{o\in O'}p'(i',o)=1$, (ii) for each $o' \in O'$, $\sum_{i\in N'}p'(i,o')=1$, and (iii) for all $i'\in N'$, $i\in N$, $o' \in O'$, and $o\in O$, [if $p(i',o)=1$ and $p'(i',o')=1$, then $o' \succsim_{i'}o$] and [if $p(i,o')=1$ and $p'(i',o')=1$, then $i'\succsim_{o'}i$]. It is well known that for the base model (and its extension with strict preferences / priorities), the core equals the set of stable deterministic matchings.

For random matchings, one can extend preferences / priorities over objects / agents to random matches via von Neumann-Morgenstern (vNM) utilities or the (incomplete) first order stochastic dominance extension. \citet{Manj13a} studies various extensions of stability and the core from deterministic to random matchings using vNM utilities, stochastic dominance requiring comparability, and stochastic dominance without requiring comparability. \citet{Manj13a} points out that for strict preferences / priorities an ex-post stable random matching is a ``weak stochastic dominance core matching'' \citep[][Proposition~3]{Manj13a}. The same observation also follows from Theorem~2 of \citet{DoYi16a}.

We prove in Appendix~\ref{appendixstochstab} two new results that clarify the relation of strong and weak dominance stability as defined in \citet{Manj13a} with some of our stability properties: (1) a random matching is a ``strong stochastic dominance stable matching'' if and only if it is ex-ante stable and (2) a random matching that is claimwise stable is a ``weak stochastic dominance stable matching''.

Our approach is complementary to that of \citet{Manj13a} in that we focus on existing stability concepts with some focus on their underlying linear programming origin and possible fairness interpretations while he introduces new stability and core concepts based on how preferences / priorities are extended from deterministic to random matchings.\hfill~$\diamond$
\end{remark}

\subsection{Relations between stability concepts}\label{section:taxonomystability}
We now provide a complete taxonomy of the stability concepts for random matchings we have introduced (see Figure~\ref{fig:part-relations}).%\medskip

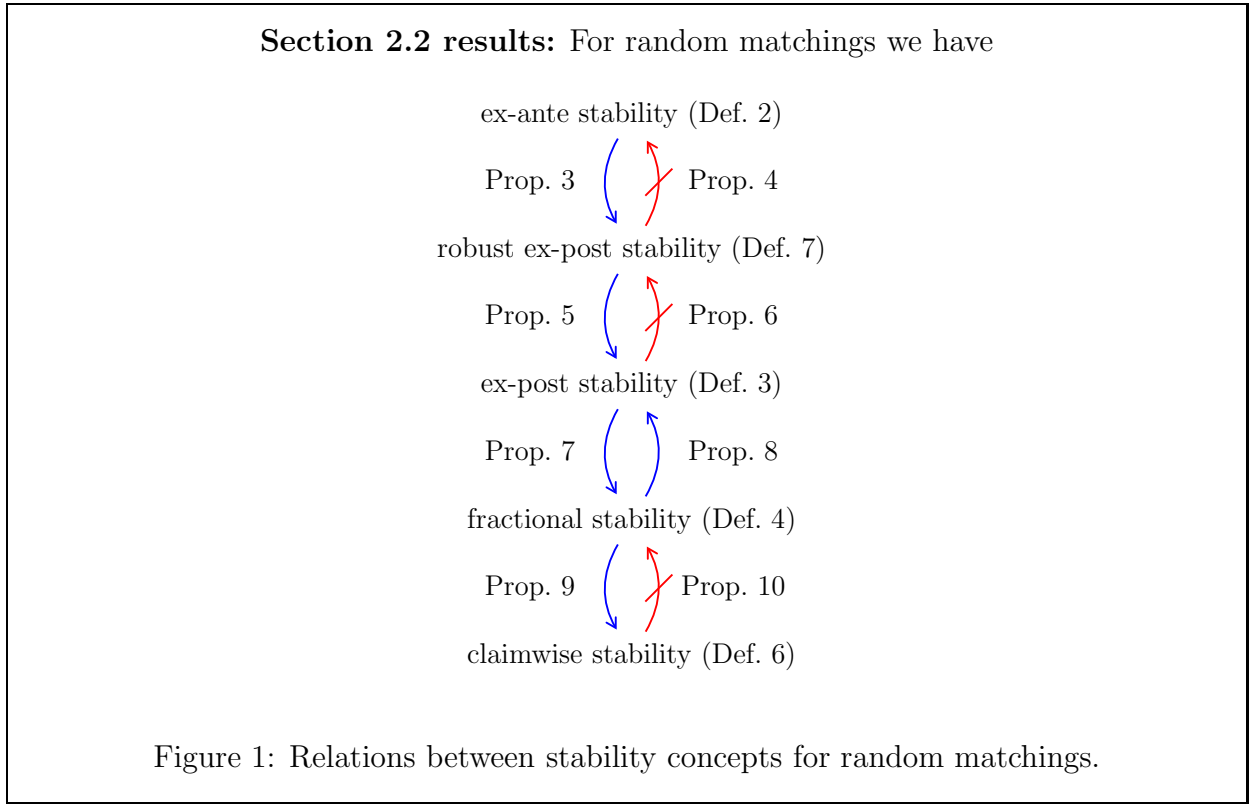
\begin{figure}[htb]
\begin{framed}
\begin{center}
\textbf{Section~\ref{section:taxonomystability} results:} For random matchings we have\\\bigskip
\scalebox{0.9}{
\begin{tikzpicture}
\tikzstyle{pfeil}=[->,>=angle 60, shorten >=1pt,draw]
\tikzstyle{onlytext}=[]
    \node[onlytext] (exante) at (2,4) {ex-ante stability (Def.~\ref{def:ex-antestability})};
	\node[onlytext] (rexpost) at (2,2) {robust ex-post stability (Def.~\ref{def:robustexpoststability})};
	\node[onlytext] (expost) at (2,0) {ex-post stability (Def.~\ref{def:expoststability})};
	\node[onlytext] (fractional) at (2,-2) {fractional  stability (Def.~\ref{def:fractionalstability})};
	\node[onlytext] (claimwise) at (2,-4) {claimwise stability (Def.~\ref{def:claimwisestable})};

    \node[onlytext] (prop2) at (0.5,3) {Prop.~\ref{prop:exante-to-rexpost}};
    \node[onlytext] (prop4) at (0.5,1) {Prop.~\ref{prop:rexpost-to-expost}};
    \node[onlytext] (prop6) at (0.5,-1) {Prop.~\ref{prop:expost-to-fractional}};
    \node[onlytext] (prop8) at (0.5,-3) {Prop.~\ref{prop:fractional-to-claimwise}};

    \node[onlytext] (prop2) at (3.5,3) {Prop.~\ref{prop:rexpost-notto-exante}};
    \node[onlytext] (prop4) at (3.5,1) {Prop.~\ref{prop:expost-notto-rexpost}};
    \node[onlytext] (prop6) at (3.5,-1) {Prop.~\ref{prop:fractional-to-expost}};
    \node[onlytext] (prop8) at (3.5,-3) {Prop.~\ref{prop:claimwise-notto-expost}};

	\draw[pfeil, thick, blue, bend right] (exante) to (rexpost);
	\draw[pfeil, thick, red, bend right] (rexpost) to (exante);
	\draw[pfeil, thick, blue, bend right] (rexpost) to (expost);
	\draw[pfeil, thick, red, bend right] (expost) to (rexpost);
	\draw[pfeil, thick, blue, bend right] (fractional) to (expost);
	\draw[pfeil, thick, blue, bend right] (expost) to (fractional);
	\draw[pfeil, thick, blue, bend right] (fractional) to (claimwise);
	\draw[pfeil, thick, red, bend right] (claimwise) to (fractional);
	\draw[thick, red] (2.2,-3.2) to (2.6,-2.8);
	
	\draw[thick, red] (2.2,0.8) to (2.6,1.2);
	
	\draw[thick, red] (2.2,2.8) to (2.6,3.2);

\end{tikzpicture}}
\end{center}	
\caption{Relations between stability concepts for random matchings.}
\label{fig:part-relations}
\end{framed}
\end{figure}

\begin{proposition}\label{prop:exante-to-rexpost}
Ex-ante stability implies robust ex-post stability.			
\end{proposition}
\begin{proof}[\textbf{Proof}]
Consider a random matching $p$ that is not robust ex-post stable. This means that $p$ can be decomposed into deterministic matchings such that one of them is not stable. Let $q$ be such an unstable deterministic matching. Since $q$ is unstable, there exist agents $i,j\in N$ and objects $o,o'\in O$ such that $q(i,o')=1$, $q(j,o)=1$, $o\succ_i o'$, and $i\succ_o j$. Since $q$ is part of a decomposition of $p$ (with positive weight), it follows that then $p(i,o')>0$, $p(j,o)>0$, $o\succ_i o'$, and  $i\succ_o j$. Hence, $p$ is not ex-ante stable.		
\end{proof}

The following example shows that even if a random matching is ex-ante stable (and hence robust ex-post stable), the decomposition into stable deterministic matchings need not be unique.

\begin{example}\label{example:ex-ante-differentdecomp}\normalfont Let $N=\{1,2,3,4\}$ and $O=\{w,x,y,z\}$. Consider the following preferences and priorities:
\begin{center}
\begin{tabular}{lcccc}
$\succ_1$:&$w$&$x$&$y$&$z$\\
$\succ_2$:&$x$&$w$&$z$&$y$\\
$\succ_3$:&$y$&$z$&$w$&$x$\\
$\succ_4$:&$z$&$y$&$x$&$w$
\end{tabular}
\quad\quad
\begin{tabular}{lcccc}
$\succ_w$:&$2$&$1$&$4$&$3$\\
$\succ_x$:&$1$&$2$&$3$&$4$\\
$\succ_y$:&$4$&$3$&$2$&$1$\\
$\succ_z$:&$3$&$4$&$1$&$2$
\end{tabular}	
\end{center}
There are four deterministic stable matchings:
\begin{center}
\begin{tabular}{c}
$p^1=\begin{pmatrix}
1&0&0&0\\
0&1&0&0\\
0&0&1&0\\
0&0&0&1
\end{pmatrix}$
\end{tabular}
\quad\quad	
\begin{tabular}{c}
$p^2=\begin{pmatrix}
1&0&0&0\\
0&1&0&0\\
0&0&0&1\\
0&0&1&0
\end{pmatrix}$
\end{tabular}
\end{center}

\begin{center}
\begin{tabular}{c}
$p^3=\begin{pmatrix}
0&1&0&0\\
1&0&0&0\\
0&0&1&0\\
0&0&0&1
\end{pmatrix}$
\end{tabular}
\quad\quad	
\begin{tabular}{c}
$p^4=\begin{pmatrix}
0&1&0&0\\
1&0&0&0\\
0&0&0&1\\
0&0&1&0
\end{pmatrix}$
\end{tabular}
\end{center}		
It is easy to check that the following random matching is ex-ante and hence also robust ex-post stable:
\begin{center}
\begin{tabular}{c}
$q=\begin{pmatrix}
\nicefrac{1}{2}&\nicefrac{1}{2}&0&0\\
\nicefrac{1}{2}&\nicefrac{1}{2}&0&0\\
0&0&\nicefrac{1}{2}&\nicefrac{1}{2}\\
0&0&\nicefrac{1}{2}&\nicefrac{1}{2}
\end{pmatrix}.$
\end{tabular}										
\end{center}	
There exist exactly two decompositions of $q$ into (stable) deterministic matchings: $$q=\frac{1}{2}p^{1}+\frac{1}{2}p^{4}=\frac{1}{2}p^{2}+\frac{1}{2}p^{3}.$$
Hence, the decomposition of $q$ into stable deterministic matchings is not unique.\hfill~$\diamond$
\end{example}

\begin{proposition}\label{prop:rexpost-notto-exante}
Robust ex-post stability does not imply ex-ante stability.
\end{proposition}
\begin{proof}[\textbf{Proof}]
Let $N=\{1,2,3\}$ and $O=\{x,y,z\}$. Consider the following preferences and priorities; they are the same as in \citet[][Example~2]{RothetalMOR93} but we use them to prove a different statement:
\begin{center}
\begin{tabular}{lccc}
$\succ_1$:&$x$&$y$&$z$\\
$\succ_2$:&$y$&$z$&$x$\\
$\succ_3$:&$z$&$x$&$y$
\end{tabular}
\quad\quad
\begin{tabular}{lccc}
$\succ_x$:&$2$&$3$&$1$\\
$\succ_y$:&$3$&$1$&$2$\\
$\succ_z$:&$1$&$2$&$3$
\end{tabular}	
\end{center}				
Then, consider $p^A$, which is the deterministic agent optimal stable matching,\footnote{The deterministic agent optimal stable matching can be computed by using the agent proposing deferred-acceptance algorithm \citep{GS62}.}	
\begin{center}
\begin{tabular}{c}
$p^A=\begin{pmatrix}
1&0&0\\
0&1&0\\
0&0&1
\end{pmatrix}$
\end{tabular}
\quad\quad	
\begin{tabular}{lccc}
$\succ_1$:&$\bm{x}$&$y$&$z$\\
$\succ_2$:&$\bm{y}$&$z$&$x$\\
$\succ_3$:&$\bm{z}$&$x$&$y$
\end{tabular}
\end{center}	
and consider $p^O$, which is the deterministic object optimal stable matching,\footnote{The deterministic object optimal stable matching can be computed by using the object proposing deferred-acceptance algorithm \citep{GS62}.}		
\begin{center}
\begin{tabular}{c}
$p^O=\begin{pmatrix}
0&0&1\\
1&0&0\\
0&1&0
\end{pmatrix}$
\end{tabular}
\quad\quad	
\begin{tabular}{lccc}
$\succ_1$:&$x$&$y$&$\bm{z}$\\
$\succ_2$:&$y$&$z$&$\bm{x}$\\
$\succ_3$:&$z$&$x$&$\bm{y}$.
\end{tabular}
\end{center}					
Let $q=\frac{1}{2}p^A+\frac{1}{2}p^O$. Thus,
\begin{center}
\begin{tabular}{c}
$q=\begin{pmatrix}
\nicefrac{1}{2}&0&\nicefrac{1}{2}\\
\nicefrac{1}{2}&\nicefrac{1}{2}&0\\
0&\nicefrac{1}{2}&\nicefrac{1}{2}
\end{pmatrix}.$
\end{tabular}										
\end{center}								
We first show that $q$'s only decomposition into deterministic matchings is the one with respect to $p^A$ and $p^O$: if the decomposition involves a deterministic matching in which agent $1$ gets object $x$, then the only deterministic matching consistent with $q$ is $p^A$ (because $q(2,z)=0$); if the decomposition involves a deterministic matching in which agent $1$ gets object $z$, then the only deterministic matching consistent with $q$ is  $p^O$ (because $q(3,x)=0$); since $q(1,y)=0$, no deterministic matching consistent with $q$ allows for agent 1 to get object $y$. Hence, we have proven that a convex decomposition of $q$ can only involve deterministic matchings $p^A$ and $p^O$.
Since both $p^A$ and $p^O$ are stable, it follows that $q$ is robust ex-post stable.\medskip
			
Second, we show that $q$ is not ex-ante stable. Note that for agents $1,2\in N$ and objects $z,y\in O$ we have that $q(1,z)>0$, $q(2,y)>0$,  $1\succ_y 2$, and $y\succ_1 z$, i.e., agent 1 ex-ante envies agent 2 for his probability share of object $y$. Hence, $q$ is not ex-ante stable.

Thus, $q$ is robust ex-post stable but not ex-ante stable.
\end{proof}
	
\begin{proposition}\label{prop:rexpost-to-expost}
Robust ex-post stability implies ex-post stability.
\end{proposition}
\begin{proof}[\textbf{Proof}]
By definition, if all decompositions of the random matching involve deterministic stable matchings, then there exists at least one decomposition that involves only deterministic stable matchings.
\end{proof}
			
\begin{proposition}\label{prop:expost-notto-rexpost}
Ex-post stability does not imply robust ex-post stability.
\end{proposition}
\begin{proof}[\textbf{Proof}]Our example and proof is the same as in \citet[][Example~2]{RothetalMOR93}. Let $N=\{1,2,3\}$ and $O=\{x,y,z\}$. Consider the following preferences and priorities: 	
\begin{center}
\begin{tabular}{lccc}
$\succ_1$:&$x$&$y$&$z$\\
$\succ_2$:&$y$&$z$&$x$\\
$\succ_3$:&$z$&$x$&$y$
\end{tabular}
\quad\quad
\begin{tabular}{lccc}
$\succ_x$:&$2$&$3$&$1$\\
$\succ_y$:&$3$&$1$&$2$\\
$\succ_z$:&$1$&$2$&$3$
\end{tabular}	
\end{center}				
Then, consider $p^A$, which is the deterministic agent optimal stable matching,	
\begin{center}
\begin{tabular}{c}
$p^A=\begin{pmatrix}
1&0&0\\
0&1&0\\
0&0&1
\end{pmatrix}$
\end{tabular}
\quad\quad	
\begin{tabular}{lccc}
$\succ_1$:&$\bm{x}$&$y$&$z$\\
$\succ_2$:&$\bm{y}$&$z$&$x$\\
$\succ_3$:&$\bm{z}$&$x$&$y$
\end{tabular}
\end{center}	
and consider $p^O$, which is the deterministic object optimal stable matching,	
\begin{center}
\begin{tabular}{c}
$p^O=\begin{pmatrix}
0&0&1\\
1&0&0\\
0&1&0
\end{pmatrix}$
\end{tabular}
\quad\quad	
\begin{tabular}{lccc}
$\succ_1$:&$x$&$y$&$\bm{z}$\\
$\succ_2$:&$y$&$z$&$\bm{x}$\\
$\succ_3$:&$z$&$x$&$\bm{y}$.
\end{tabular}
\end{center}
The only other deterministic stable matching is
	\begin{center}
\begin{tabular}{c}
$p=\begin{pmatrix}
0&1&0\\
0&0&1\\
1&0&0
\end{pmatrix}$
\end{tabular}
\quad\quad	
\begin{tabular}{lccc}
$\succ_1$:&$x$&$\bm{y}$&$z$\\
$\succ_2$:&$y$&$\bm{z}$&$x$\\
$\succ_3$:&$z$&$\bm{x}$&$y$.
\end{tabular}
\end{center}	
Let $q$ be the uniform random matching. Thus,
\begin{center}
\begin{tabular}{c}
$q=\begin{pmatrix}
\nicefrac{1}{3}&\nicefrac{1}{3}&\nicefrac{1}{3}\\
\nicefrac{1}{3}&\nicefrac{1}{3}&\nicefrac{1}{3}\\
\nicefrac{1}{3}&\nicefrac{1}{3}&\nicefrac{1}{3}
\end{pmatrix}.$
\end{tabular}										
\end{center}								
Note that
\[q=
\frac{1}{3}
\begin{pmatrix}
1&0&0\\
0&1&0\\
0&0&1
\end{pmatrix}
+\frac{1}{3}
\begin{pmatrix}
0&0&1\\
1&0&0\\
0&1&0
\end{pmatrix}
+\frac{1}{3}
\begin{pmatrix}
0&1&0\\
0&0&1\\
1&0&0
\end{pmatrix}			=\frac{1}{3}p^A+\frac{1}{3}p^O+\frac{1}{3}p.\]
Since $q$ can be decomposed into deterministic stable matchings, it is ex-post stable.\medskip
			
We now show that the uniform random matching $q$ is not robust ex-post stable. Note that
\[q=
\frac{1}{3}
\begin{pmatrix}
0&1&0\\
1&0&0\\
0&0&1
\end{pmatrix}
+\frac{1}{3}
\begin{pmatrix}
0&0&1\\
0&1&0\\
1&0&0
\end{pmatrix}
+\frac{1}{3}
\begin{pmatrix}
1&0&0\\
0&0&1\\
0&1&0
\end{pmatrix}\]
where all the deterministic matchings in the decomposition are unstable. Hence, $q$ is not robust ex-post stable.

Thus, $q$ is ex-post stable but not robust ex-post stable.
\end{proof}															

Next, as already mentioned when introducing fractional stability, fractional stability is equivalent to ex-post stability \citep[see][]{RothetalMOR93,TeSe98a}. This equivalence is based on the insight by \cite{Vand89a} that both stability concepts are convex with deterministic stable matchings as extreme points (we show in Proposition~\ref{prop:weakfractional-notto-weakexpost} that once preferences and priorities can be weak, this statement isn't correct anymore for the convex set of fractionally weakly stable random matchings). We add the proofs for completeness.
			
\begin{proposition}\label{prop:expost-to-fractional}
Ex-post stability implies fractional stability.
\end{proposition}
\begin{proof}[\textbf{Proof}]
If a random matching is ex-post stable then by definition it can be written as a convex combination of deterministic stable matchings. All of these deterministic stable matchings are fractionally stable. Since the set of fractionally stable matchings is convex (Lemma~\ref{lemma:frac-convex}), a convex combination of deterministic stable matchings is fractionally stable.
\end{proof}
					
\begin{proposition}\label{prop:fractional-to-expost}
Fractional stability implies ex-post stability.
\end{proposition}
\begin{proof}[\textbf{Proof}]
As already mentioned when introducing fractional stability, for strict priorities, the extreme points of the polytope defined by the linear inequalities (\ref{LE1}), (\ref{LE2}), (\ref{LE3}), and (\ref{LE4}) are exactly the (incidence vectors of the) deterministically stable matchings \citep{Vand89a,Rothblum1992,RothetalMOR93}. Since, by definition, fractionally stable random matchings are solutions to the linear inequalities, a fractionally stable random matching can be decomposed into deterministic stable matchings, which implies that a fractionally stable random matching is ex-post stable.
\end{proof}

\begin{proposition}\label{prop:fractional-to-claimwise}
Fractional stability implies claimwise stability.
\end{proposition}
\begin{proof}[\textbf{Proof}]
Consider a random matching $p$ that is not claimwise stable. Then, for some pair $(i,o)\in N\times O$ and some $j\in N$ such that $i\succ_oj$, strict inequality (\ref{LE7}) applies: $$p(j,o)>\sum_{o':o'\succ_i o}p(i,o'),$$ i.e., agent $i$ has a claim against agent $j$ with respect to object $o$.
But this implies that $$\sum_{k:k\prec_o i}p(k,o)> \sum_{o':o'\succ_i o}p(i,o').$$ Hence, $p$ is not fractionally stable.
Thus, fractional stability implies claimwise stability.
\end{proof}

\begin{proposition}\label{prop:claimwise-notto-expost}
Claimwise stability does not imply ex-post / fractional stability.
\end{proposition}	
\begin{proof}[\textbf{Proof}]
Let $N=\{1,2,3\}$ and $O=\{x,y,z\}$. Consider the following preferences and priorities: 	
\begin{center}
\begin{tabular}{lccc}
$\succ_1$:&$x$&$z$&$y$\\
$\succ_2$:&$y$&$x$&$z$\\
$\succ_3$:&$z$&$x$&$y$
\end{tabular}
\quad\quad
\begin{tabular}{lccc}
$\succ_x$:&$2$&$3$&$1$\\
$\succ_y$:&$1$&$3$&$2$\\
$\succ_z$:&$2$&$1$&$3$
\end{tabular}	
\end{center}				
Then, consider $p^A$, which is the deterministic agent optimal stable matching,	
\begin{center}
\begin{tabular}{c}
$p^A=\begin{pmatrix}
1&0&0\\
0&1&0\\
0&0&1
\end{pmatrix}$
\end{tabular}
\quad\quad	
\begin{tabular}{lccc}
$\succ_1$:&$\bm{x}$&$z$&$y$\\
$\succ_2$:&$\bm{y}$&$x$&$z$\\
$\succ_3$:&$\bm{z}$&$x$&$y$
\end{tabular}
\end{center}	
and consider $p^O$, which is the deterministic object optimal stable matching,	
\begin{center}
\begin{tabular}{c}
$p^O=\begin{pmatrix}
0&0&1\\
1&0&0\\
0&1&0
\end{pmatrix}$
\end{tabular}
\quad\quad	
\begin{tabular}{lccc}
$\succ_1$:&$x$&$\bm{z}$&$y$\\
$\succ_2$:&$y$&$\bm{x}$&$z$\\
$\succ_3$:&$z$&$x$&$\bm{y}$.
\end{tabular}
\end{center}
Let $q$ be the uniform random matching. Thus,
\begin{center}
\begin{tabular}{c}
$q=\begin{pmatrix}
\nicefrac{1}{3}&\nicefrac{1}{3}&\nicefrac{1}{3}\\
\nicefrac{1}{3}&\nicefrac{1}{3}&\nicefrac{1}{3}\\
\nicefrac{1}{3}&\nicefrac{1}{3}&\nicefrac{1}{3}
\end{pmatrix}.$
\end{tabular}										
\end{center}	
First, since $p^O$ is the deterministic object optimal stable matching, agent 1 does not get $y$ in any deterministic stable matching. Hence, random matching $q$ is not ex-post stable. Alternatively, we can check that fractional stability is violated and inequality (\ref{LE6}) holds for agent 2 and object $x$:
$$\frac{2}{3}=\sum_{j:j\prec_x 2}q(j,x)>\sum_{o':o'\succ_2 x}q(2,o') =\frac{1}{3} .$$
Second, we show that random matching $q$ is claimwise stable by checking if there are claims of an agent $i$ against an agent $j$, i.e., are there $(i,o)\in N\times O$ and $j\in N$ such that $i\succ_oj$ and $q(j,o)>\sum_{o':o'\succ_i o}q(i,o')$? We show that there are no claims.
\begin{itemize}
\item For an agent $i\in N$, a claim for a higher probability for his most preferred object against any of the other agents is not justified because all other agents have higher priority for that object.
\item For an agent $i\in N$, a claim for a higher probability for his second preferred object against any of the other agents is not justified because he gets an object in the strict upper contour set of his second preferred object with probability $\nicefrac{1}{3}$ whereas any other agent also gets that object with probability $\nicefrac{1}{3}$ (a probability that is not higher).
\item No agent $i\in N$ would claim a higher probability for his least preferred object (because he gets an object in the strict upper contour set of his least preferred object with probability $\nicefrac{2}{3}$ whereas any other agent only gets that object with probability $\nicefrac{1}{3}$).
\end{itemize}				
\end{proof}

We conclude with a proof of Proposition~\ref{proposition1}. We show that for deterministic matchings, all the stability concepts for random matchings coincide with stability for deterministic matchings.

\begin{proof}[\textbf{Proof of Proposition~\ref{proposition1}}]
Let deterministic matching $p$ be stable and note that for deterministic matchings no envy implies ex-ante stability \citep[as also noted by][]{KeUn15a}. By Proposition~\ref{prop:exante-to-rexpost}, $p$ is robust ex-post stable; by Proposition~\ref{prop:rexpost-to-expost}, $p$ is ex-post stable; by Proposition~\ref{prop:expost-to-fractional}, $p$ is fractionally stable; and by Proposition~\ref{prop:fractional-to-claimwise}, $p$ is claimwise stable. We are done if we can show that any deterministic claimwise stable matching is stable. Assume, by contradiction, that there exists a deterministic matching $q$ that is claimwise stable but not stable. Then, there exist $i,j\in N$ and $o,o'\in O$ such that $q(i,o')=1$, $q(j,o)=1$, $o\succ_i o'$, and $i\succ_{o}j$. But then, $1=q(j,o)>\sum_{o'':o''\succ_i o}p(i,o'')=0$ and agent $i$ has a claim against agent $j$. \citet{Afac15a} also proved that any deterministic claimwise stable matching is stable.
\end{proof}

We next extend our base model and corresponding results in two steps. First adding weak preferences and weak priorities (Section~\ref{subsection:weakprefpri}) and second allowing for unacceptability and a different number of agents and objects (Section~\ref{subsection:weakdifferentnumbers}) allows us to separately show the required adjustments in the stability concepts and the associated proof techniques required when (stepwise) extending the base model.

\subsection{Weak preferences and weak priorities}\label{subsection:weakprefpri}

Recall that with strict preferences and strict priorities, a deterministic matching $p$ is stable if for each pair $(i,o)\in N\times O$,
\begin{equation}\tag{\ref{LE4}}
p(i,o)+\sum_{o':o'\succ_i o}p(i,o')+ \sum_{j:j\succ_{o}i}p(j,o)\geq 1.
\end{equation}

If preferences or priorities can be weak, i.e., agents' preferences are weak orders over $O$ and objects' priorities are weak orders over $N$, then various deterministic stability notions with varying degrees of strength are possible \citep[see][]{IrvingDAM1994}. A deterministic matching $p$ is \emph{weakly stable} if there is no \emph{strict blocking} agent-object pair such that, by being matched, each would be strictly better off than at their current matches at $p$.  In the case of \emph{strong stability}, there is no \emph{weak blocking} agent-object pair such that, by being matched, one of them is strictly better off, whilst the other must be no worse off than at their current matches at $p$.  While for weak preferences and weak priorities, weakly stable deterministic matchings always exist, it is well known that the set of strongly stable deterministic matchings may be empty. Furthermore, we consider the absence of strict blocking pairs as the most natural no envy / stability notion and therefore focus on weak stability.

\begin{definition}[\textbf{No envy / weak stability for deterministic matchings}]
\normalfont A deterministic matching $p$ has \textbf{\emph{no envy}} or is \textbf{\emph{weakly stable}} if there exists no agent $i$ who is matched to object $o'$ but prefers object $o$ while object $o$ is matched to some agent $j$ with lower priority than $i$, i.e., there exist no $i,j\in N$ and no $o,o'\in O$ such that $p(i,o')=1$, $p(j,o)=1$, $o\succ_i o'$, and $i\succ_{o}j$.\label{def:weakstability}
\end{definition}

A deterministic matching $p$ is weakly stable if it satisfies the following inequalities: for each pair $(i,o)\in N\times O$,
\begin{equation}\label{LE4W}
p(i,o)+\sum_{o':o'\succsim_i o;o'\neq o}p(i,o')+ \sum_{j:j\succsim_{o}i;j\neq i}p(j,o)\geq 1.
\end{equation}

If one breaks all preference and priority ties, then the well-known deferred-acceptance algorithm \citep{GS62} computes a deterministic matching that is weakly stable.\medskip

The definitions of ex-ante, ex-post, and robust ex-post stability essentially remain the same as before.

\begin{definition}[\textbf{No ex-ante envy / ex-ante weak stability}]
\normalfont	A random matching $p$ has \emph{\textbf{no ex-ante envy}} or is \emph{\textbf{ex-ante weakly stable}} if there exists no agent $i$ who is matched with positive probability to object $o'$ but prefers a higher probability for object $o$ while object $o$ is matched with positive probability to some agent $j$ with lower priority than $i$, i.e., there exists no $i,j\in N$ and no $o,o'\in O$ such that $p(i,o')>0$, $p(j,o)>0$, $o\succ_i o'$, and $i\succ_oj$.\label{def:weakex-antestability}
\end{definition}

Note that the definition of Aharoni-Fleiner fractional stability (Definition~\ref{def:AFfractionalstability}) remains the same and its equivalence to no ex-ante envy follows as before.

\begin{definition}[\textbf{Ex-post weak stability}]
\normalfont	A random matching $p$ is \emph{\textbf{ex-post weakly stable}} if it can be decomposed into deterministic weakly stable matchings.\label{def:weakexpoststability}
\end{definition}

\begin{definition}[\textbf{Robust ex-post weak stability}]
\label{def:weakrobustexpoststability}
\normalfont A random matching $p$ is \emph{\textbf{robust ex-post weakly stable}} if all its decompositions are into deterministic weakly stable matchings.
\end{definition}

Next, the definition of deterministic weak stability leads to the following associated stability concept (by relaxing the ``integer solution requirement'' for inequalities (\ref{LE4W})).

\begin{definition}[\textbf{Fractional weak stability and violations of fractional weak stability}]
\normalfont	A random matching $p$ is \textit{\textbf{fractionally weakly stable}} if for each pair $(i,o)\in N\times O$,
\begin{equation}\tag{\ref{LE4W}}
p(i,o)+\sum_{o':o'\succsim_i o;o'\neq o}p(i,o')+ \sum_{j:j\succsim_{o}i;j\neq i}p(j,o)\geq 1,
\end{equation}or more compactly,
\begin{equation}\label{LEB}
\sum_{o':o'\succsim_i o;o'\neq o}p(i,o') \geq  \sum_{j:j\prec_o i}p(j,o).
\end{equation}
A \textit{\textbf{violation of fractional weak stability occurs if there exists a} pair $(i,o)\in N\times O$ such that}
\begin{equation}\label{LEC}
\sum_{j:j\prec_o i}p(j,o)>\sum_{o':o'\succsim_i o;o'\neq o}p(i,o').
\end{equation}\label{def:weakfractionalstability}
\end{definition}

Inequality (\ref{LEC}) implies $\sum_{o':o'\succsim_i o}p(i,o')<1$, i.e., agent $i$ receives some fraction of an object in his strict lower contour set at $o$ (if not, this would imply that $\sum_{o':o'\succsim_i o;o'\neq o}p(i,o')+p(i,o)=1$ and hence, $\sum_{j:j\prec_o i}p(j,o)+p(i,o)>1$; a contradiction to feasibility). Thus, agent $i$ would want to consume more of object $o$. Inequality (\ref{LEC}) also implies $\sum_{j:j\succsim_o i}p(j,o)<1$, i.e., object $o$ receives some fraction of an agent in its strict lower contour set at $i$. Thus, object $o$ would want to consume more of agent $i$. Moreover, strict inequality (\ref{LEC}) encodes the following envy notion: using consumption process language, as long as agent $i$ consumes objects that are different from and not worse than $o$ he does not envy the set of lower priority agents to jointly consume fractions of $o$, however, once the set of lower priority agents have consumed as much of $o$ as agent $i$'s weak upper contour set at $o$ (not including $o$), agent $i$ starts having envy towards them for any additional amounts of $o$ (unless agent $i$ can fill his remaining probability quota with object $o$).

\begin{remark}[\textbf{A symmetric reformulation of fractional weak stability and its violations}]\label{remark:symmetricFractionalWeak}
\normalfont In the definition of fractional weak stability by inequalities (\ref{LEB}) and of a violation of fractional weak stability by inequality (\ref{LEC}) we have taken the viewpoint of an agent who considers the consumptions of lower priority agents for an object. The symmetric formulations when taking the viewpoint of an object that ``considers'' the matches of lower preferred objects to an agent are as follows. A random matching $p$ is fractionally weakly stable if for each pair $(i,o)\in N\times O$,
\begin{equation}\tag{\ref{LEB}'}
\sum_{j:j\succsim_o i;j\neq i}p(j,o) \geq  \sum_{o':o'\prec_i o}p(i,o').
\end{equation}
We can write a violation of fractional weak stability as, there exists a pair $(i,o)\in N\times O$ such that
\begin{equation}\tag{\ref{LEC}'}
\sum_{o':o'\prec_i o}p(i,o')>\sum_{j:j\succsim_o i;j\neq i}p(j,o).
\end{equation}\hfill~$\diamond$
\end{remark}

The following lemma follows from the definition of fractional weak stability via linear inequalities.

\begin{lemma}\label{lemma:weakfrac-convex}
Fractional weak stability is convex.
\end{lemma}

When preferences can be weak, then the notion of a claim can be adjusted as follows: using consumption process language, as long as agent $i$ consumes objects that are different from and not worse than $o$ he does not envy lower priority agent $j$ to consume fractions of $o$, however, once agent $j$ has consumed as much of $o$ as agent $i$'s weak upper contour set at $o$ (not including $o$), agent $i$ starts having envy towards him for any additional amounts of $o$ (unless agent $i$ can fill his remaining probability quota with object $o$). An agent $i\in N$ has a \emph{claim} against an agent $j\in N$, if there exists an object $o\in O$ such that $i\succ_o j$ and
\begin{equation}\label{LE7New}
p(j,o)>\sum_{o
:o'\succsim_i o;o'\neq o}p(i,o').
\end{equation}
Inequality (\ref{LE7New}) implies $\sum_{o':o'\succsim_i o}p(i,o')<1$, i.e., agent $i$ receives some fraction of an object in his strict lower contour set at $o$ (if not, this would imply that $\sum_{o':o'\succsim_i o;o'\neq o}p(i,o')+p(i,o)=1$ and hence, $p(j,o)+p(i,o)>1$; a contradiction to feasibility). Thus, agent $i$ would want to consume more of object $o$.\medskip

A random matching is \emph{claimwise weakly stable} if it does not admit any claim.

\begin{definition}[\textbf{Claimwise weak stability}]
\normalfont	A random matching $p$ is \textit{\textbf{claimwise weakly stable}} if for each pair $(i,o)\in N\times O$  and each $j\in N$ such that $i\succ_oj$,
\begin{equation}\label{LE8New}
\sum_{o':o'\succsim_i o;o'\neq o}p(i,o')\geq  p(j,o).
\end{equation}\label{def:weaklyclaimwisestable}
\end{definition}

The following lemma follows from the definition of claimwise weak stability via linear inequalities.

\begin{lemma}\label{lemma:weakclaimwise-convex}
Claimwise weak stability is convex.
\end{lemma}

It follows easily that if we restrict attention to deterministic matchings, then all the weak stability concepts for random matchings coincide with standard weak stability / no envy (Definition~\ref{def:weakstability}). The proof of Proposition~\ref{proposition1weak} follows the same arguments as the proof of our previous Proposition~\ref{proposition1} and we therefore omit it.
	
\begin{proposition}\label{proposition1weak}
For deterministic matchings, all the weak stability concepts for random matchings with weak preferences and weak priorities coincide with weak stability / no envy for deterministic matchings.
\end{proposition}

Our taxonomy of the stability concepts for random matchings  with weak preferences and weak priorities now looks as follows (see Figure~\ref{fig:weakpart-relations}).
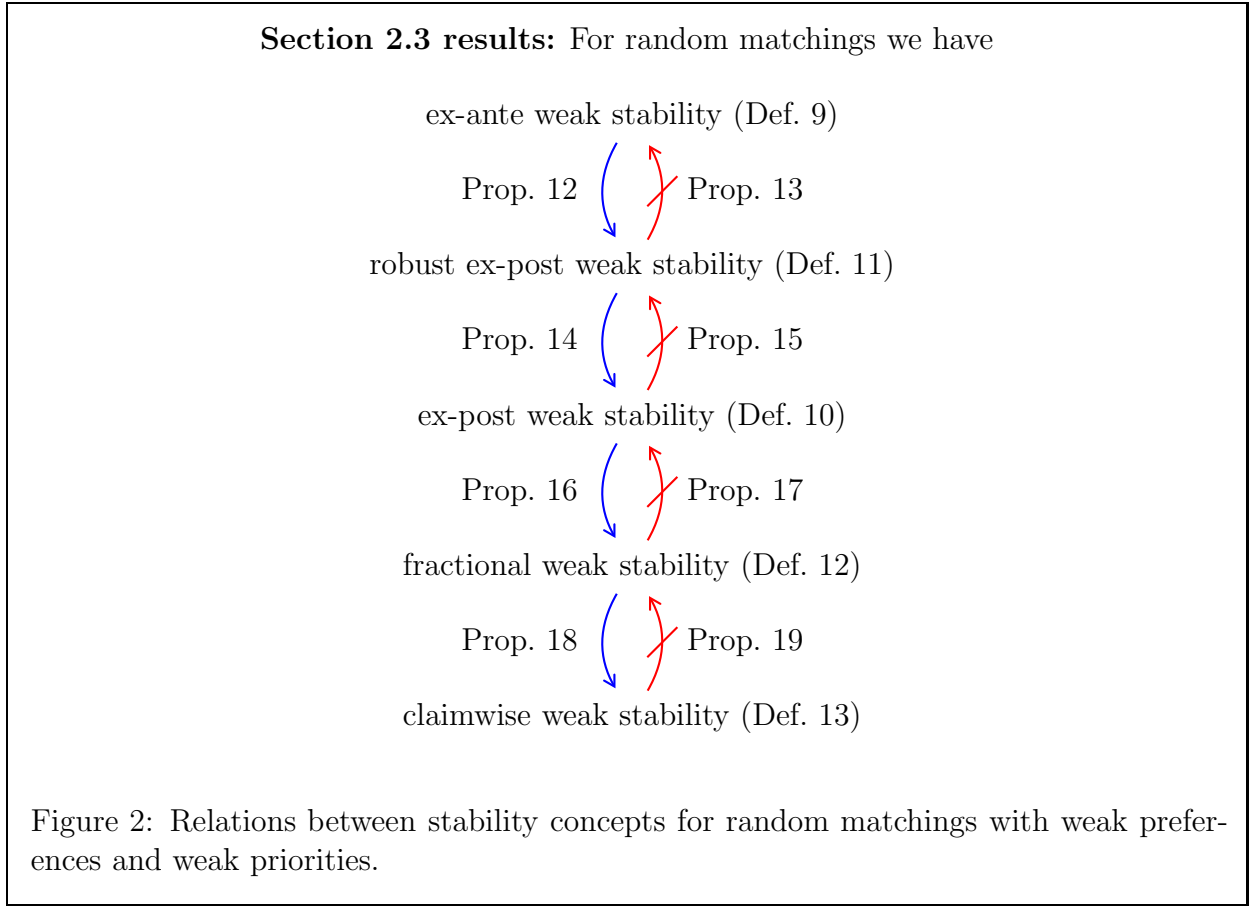
\medskip

\begin{figure}[H]%[!htbp]
\begin{framed}
\begin{center}
\textbf{Section~\ref{subsection:weakprefpri} results:} For random matchings we have\\\bigskip
\scalebox{1}{
\begin{tikzpicture}
\tikzstyle{pfeil}=[->,>=angle 60, shorten >=1pt,draw]
\tikzstyle{onlytext}=[]
    \node[onlytext] (exante) at (2,4) {ex-ante weak stability (Def.~\ref{def:weakex-antestability})};
	\node[onlytext] (rexpost) at (2,2) {robust ex-post weak stability (Def.~\ref{def:weakrobustexpoststability})};
	\node[onlytext] (fractional) at (2,-2) {fractional weak stability (Def.~\ref{def:weakfractionalstability})};
	\node[onlytext] (expost) at (2,0) {ex-post weak stability (Def.~\ref{def:weakexpoststability})};
	\node[onlytext] (claimwise) at (2,-4) {claimwise weak stability (Def.~\ref{def:weaklyclaimwisestable})};

    \node[onlytext] (prop11) at (0.5,3) {Prop.~\ref{prop:weakexante-to-rexpost}};
    \node[onlytext] (prop13) at (0.5,1) {Prop.~\ref{prop:weakrexpost-to-expost}};
    \node[onlytext] (prop15) at (0.5,-1) {Prop.~\ref{prop:weakexpost-to-fractional}};
    \node[onlytext] (prop16) at (0.5,-3) {Prop.~\ref{prop:weakfractional-to-claimwise}};

    \node[onlytext] (prop2) at (3.5,3) {Prop.~\ref{prop:weakrexpost-notto-exante}};
    \node[onlytext] (prop4) at (3.5,1) {Prop.~\ref{prop:weakexpost-notto-rexpost}};
	\node[onlytext] (open) at (3.5,-1) {Prop.~\ref{prop:weakfractional-notto-weakexpost}};
    \node[onlytext] (prop8) at (3.5,-3) {Prop.~\ref{prop:weakclaimwise-notto-expost}};
	
	\draw[pfeil, thick, blue, bend right] (rexpost) to (expost);
	\draw[pfeil, thick, blue, bend right] (fractional) to (claimwise);
	\draw[pfeil, thick, blue, bend right] (exante) to (rexpost);
	\draw[pfeil, thick, blue, bend right] (expost) to (fractional);

    \draw[pfeil, thick, red, bend right] (fractional) to (expost);
    
    \draw[pfeil, thick, red, bend right] (expost) to (rexpost);
    
	\draw[pfeil, thick, red, bend right] (claimwise) to (fractional);
	\draw[pfeil, thick, red, bend right] (rexpost) to (exante);

	\draw[thick, red] (2.2,-3.2) to (2.6,-2.8);
	
	\draw[thick, red] (2.2,-1.2) to (2.6,-0.8);
	
	\draw[thick, red] (2.2,0.8) to (2.6,1.2);
	
	\draw[thick, red] (2.2,2.8) to (2.6,3.2);
	
\end{tikzpicture}}
\end{center}	
\caption{Relations between stability concepts for random matchings with weak preferences and weak priorities.}
\label{fig:weakpart-relations}
\end{framed}
\end{figure}

The arguments in the proofs of Propositions~\ref{prop:exante-to-rexpost} and \ref{prop:rexpost-notto-exante} remain valid to prove that ex-ante weak stability implies robust ex-post weak stability but not vice versa.

\begin{proposition}\label{prop:weakexante-to-rexpost}
Ex-ante weak stability implies robust ex-post weak stability.			
\end{proposition}

\begin{proposition}\label{prop:weakrexpost-notto-exante}
Robust ex-post weak stability does not imply ex-ante weak stability.
\end{proposition}

The arguments in the proofs of Propositions~\ref{prop:rexpost-to-expost} and \ref{prop:expost-notto-rexpost} remain valid to prove that robust ex-post weak stability implies ex-post weak stability but not vice versa.

\begin{proposition}\label{prop:weakrexpost-to-expost}
Robust ex-post weak stability implies ex-post weak stability.
\end{proposition}

\begin{proposition}\label{prop:weakexpost-notto-rexpost}
Ex-post weak stability does not imply robust ex-post weak stability.
\end{proposition}

The arguments in the proof of Proposition~\ref{prop:expost-to-fractional} remain valid to prove that ex-post weak stability implies fractional weak stability.

\begin{proposition}\label{prop:weakexpost-to-fractional}
Ex-post weak stability implies fractional weak stability.
\end{proposition}

However, Proposition~\ref{prop:fractional-to-expost} does not extend to weak preferences and weak priorities. The example to prove this is due to Battal Do\u{g}an.

\begin{proposition}\label{prop:weakfractional-notto-weakexpost}
Fractional weak stability does not imply ex-post weak stability.
\end{proposition}

\begin{proof}[\textbf{Proof}]
Let $N=\{1,2,3\}$ and $O=\{x,y,z\}$. Consider the following preferences and priorities (the brackets indicate indifferences): 	
\begin{center}
\begin{tabular}{lccc}
$\succ_1$:&$[x\ y\ z]$&&\\
$\succ_2$:&$y$&$x$&$z$\\
$\succ_3$:&$[x\ y\ z]$&&
\end{tabular}
\quad\quad
\begin{tabular}{lcc}
$\succ_x$:&$[2\ 3]$&$1$\\
$\succ_y$:&$[1\ 2\ 3]$&\\
$\succ_z$:&$[1\ 2\ 3]$&
\end{tabular}	
\end{center}				
Consider random matching $q$, which is fractionally weakly stable because agents~1 and 3 only get best objects and from agent~2's perspective no agent with a lower priority consumes his best object~$y$, which he receives with probability $\frac{1}{2}$, and agent~2, who does have a lower priority for object~$x$ does not consume more that $\frac{1}{2}$ of $x$,	
\begin{center}
\begin{tabular}{c}
$q=\begin{pmatrix}
\nicefrac{1}{2}&0&\nicefrac{1}{2}\\
0&\nicefrac{1}{2}&\nicefrac{1}{2}\\
\nicefrac{1}{2}&\nicefrac{1}{2}&0
\end{pmatrix}.$
\end{tabular}										
\end{center}
Note that random matching $q$ has a unique decomposition into the deterministic matchings
\begin{center}
\begin{tabular}{c}
$p^{1}=\begin{pmatrix}
1&0&0\\
0&0&1\\
0&1&0
\end{pmatrix}$
\end{tabular}
\quad\quad
\begin{tabular}{c}
$p^{2}=\begin{pmatrix}
0&0&1\\
0&1&0\\
1&0&0
\end{pmatrix}$
\end{tabular}
\end{center}
such that $q=\frac{1}{2}p^{1}+\frac{1}{2}p^{2}$. However, deterministic matching $p^{1}$ is weakly unstable because agent~$2$ justifiably envies agent~$1$. Hence, random matching $q$ is not ex-post weakly stable.\end{proof}

The arguments in the proofs of Propositions~\ref{prop:fractional-to-claimwise} and \ref{prop:claimwise-notto-expost} remain valid to prove that fractional weak stability implies claimwise weak stability but not vice versa.

\begin{proposition}\label{prop:weakfractional-to-claimwise}
Fractional weak stability implies claimwise weak stability.
\end{proposition}

\begin{proposition}\label{prop:weakclaimwise-notto-expost}
Claimwise weak stability does not imply fractional weak stability.
\end{proposition}	
	
\section{Generalized random matchings: weak preferences, weak priorities, unacceptability, and different numbers of agents and objects}\label{section:extensions}\label{subsection:weakdifferentnumbers}

We now further generalize the model and consider the setting where there is a set of $n$ agents $N=\{1,\ldots, n\}$  (in lexicographic order $1,\ldots, n$) and a set of $m$ objects $O=\{o_1,\ldots, o_m\}$  (in lexicographic order $o_1,\ldots, o_m$) where $m$ can be less than, equal to, or more than $n$. We also now allow the agents and objects to partition the other side into acceptable and unacceptable entities. An agent / object would rather be unmatched than to be matched to an unacceptable object / agent.
As in Section~\ref{subsection:weakprefpri}, preferences and priorities can be weak.\medskip

We still assume that agents would like to consume (up to) one object, but with different numbers of agents and objects and taking acceptability into account we relax the notion of a random matching as follows. A \textit{\textbf{generalized random matching}} $p$ is a $n\times m$ matrix $[p(i,o)]_{i\in N,o\in O}$ such that
\begin{equation}\label{LE1gen}\mbox{for each pair }(i,o)\in N\times O,\ p(i,o)\geq 0.\end{equation}
\begin{equation}\label{LE2gen}\mbox{for each }i\in N,\ \sum_{o\in O}p(i,o)\leq 1,\mbox{ and }\end{equation}
\begin{equation}\label{LE3gen}\mbox{for each }o\in O,\ \sum_{i\in N}p(i,o)\leq 1.\end{equation}
Hence, a random matching is a special case of a generalized random matching.
A generalized random matching $p$ is \textbf{\textit{deterministic}} if for each pair $(i,o)\in N\times O$, $ p(i,o) \in \{0,1\}$.\medskip

Each generalized random matching can be represented as a convex combination of generalized deterministic matchings. The statement follows from the fact that every doubly substochastic matrix is a finite convex combination of partial permutation matrices~\citep[][Section~3.2, pp.~164-165]{Horn86a}. \citet{KoMa10a} also give an explicit argument for the same statement in Proposition~1 of their paper.
A \textbf{\emph{decomposition}} of a generalized random matching $p$ into generalized deterministic matchings $P_j$ ($j\in \{1,\ldots,k\}$) equals a sum $p=\sum_{j=1}^k \lambda_jP_j$ such that for each $j\in \{1,\ldots,k\}$, $\lambda_j\in (0,1]$ and $\sum_{j=1}^k\lambda_j=1$.\medskip

At a generalized deterministic matching, it can now happen that an agent gets no object at all or that an object is not assigned to any agent. We now adjust the no envy definition to take the first of these issues into account when defining no envy.

\begin{definition}[\textbf{No envy for generalized deterministic matchings}]
\normalfont	A generalized deterministic matching $p$ has  \emph{\textbf{no envy}} if there exists no agent $i$ who is matched to object $o'$ or does not receive an object but  prefers object $o$ while object $o$ is matched to some agent $j$ with lower priority than $i$, i.e., there exist no $i,j\in N$ and no $o\in O$ such that $\sum_{o':o'\succsim_i o}p(i,o')=0$ (agent $i$ does not receive $o$ or any better object), $p(j,o)=1$, and $i\succ_{o}j$.\label{def:weakgeneralnoenvy}
\end{definition}

Two additional properties for generalized random matchings will play an important role. The first is a minimal efficiency requirement that assures that no agents would rather like to obtain a higher probability for any object that isn't (fully) allocated.

\begin{definition}[\textbf{Non-wastefulness}]
\normalfont	A generalized random matching $p$ is \emph{\textbf{non-wasteful}} if there is no acceptable pair $(i,o)\in N\times O$, $\sum_{o':o'\succsim_i o}p(i,o')<1$ ($i$ would like to have more of $o$), and $\sum_{j\in N}p(j,o)<1$ ($o$ is not fully allocated).\label{def:weakgeneralnonwastefulness}
\end{definition}

Note that in our previous model, non-wastefulness was built into the definition of a random matching. Furthermore, the role of agents and objects in the definitions of no envy and non-wastefulness is not symmetric.\medskip

The second property is a voluntary participation property that ensures that no agent / object is ever matched to an unacceptable object / agent (not even partially).

\begin{definition}[\textbf{Individual rationality}]
\normalfont	A generalized random matching $p$ is \emph{\textbf{individually rational}} if for each pair $(i,o)\in N\times O$ such that at least one of $i$ and $o$ considers the other unacceptable it follows that $p(i,o)=0$.\label{def:weakgeneralindividualrationality}
\end{definition}

Since in our previous models all agents / objects were acceptable, individual rationality was automatically satisfied. We differ a bit in our terminology from \citet{WuRothWP2017} in that we do not include individual rationality in our definition of no envy.\medskip

Next we show that any ``general instance''  with an unequal number of agents and objects and with unacceptability can be transformed into an ``associated instance'' in which the number of agents and objects is equal and all entities are acceptable. The purpose of this approach is to obtain a better understanding of our general model in connection with the base model and to show how almost all our stability concepts are preserved by that transformation (in fact, all but one of the stability concepts are even equivalent under our transformation). This approach will also be crucial in establishing a taxonomy of stability concepts for the general model (see Figure~\ref{fig:weakassociated-random}).\medskip

To formalize a general instance, let the empty set $\emptyset$ symbolize the so called \textit{\textbf{null object}} which stands for being unmatched (or possibly an outside option). All objects / agents that are ranked above the null object are \textit{\textbf{acceptable}} and all objects / agents ranked below are \textit{\textbf{unacceptable}}. We now assume that agents' preferences and objects' priorities are weak orders over $O\cup\{\emptyset\}$ and $N\cup\{\emptyset\}$ respectively. Furthermore, no object is indifferent with the null object for any agent and no agent has the same priority as the null object for any object, i.e., for each pair $(i,o)\in N\times O$, [either $o\succ_i \emptyset$ or $\emptyset\succ_i o$] and [either $i\succ_o \emptyset$ or $\emptyset \succ_o i$]. A \textit{\textbf{general instance}} is denoted by the set of agents, the set of objects and the corresponding preferences and priorities: $I=(N,O,\succsim)$. If at a general instance all agents and objects are acceptable, then the null object is the least preferred entity in all preferences and priorities.\medskip

Our model does include so called \emph{\textbf{school choice}} instances as special cases: the set of agents equals a set of students, the set of objects equals a set of school seats where each school provides a fixed capacity of seats, students strictly rank schools but don't care which seat at a school they are matched to, and each school seat is allocated with the priority of the associated school \citep[recent surveys on school choice are][]{Abd13,Pat11}.\medskip

Next, we introduce a transformation from any general instance to an instance  in which the number of agents and objects is equal and all entities are acceptable and such that almost all our stability concepts are preserved / equivalent under the transformation.

\paragraph{Transforming an instance with unequal number of agents and objects and with unacceptability to one in which the number of agents and objects is equal and all entities are acceptable:}
Consider a general instance $I=(N,O,\succsim)$. Then, we can transform $I$ into an \emph{\textbf{associated instance}} $I'=(N',O',\succsim')$
of the standard setting with weak preferences and weak priorities (Section~\ref{subsection:weakprefpri}) where $|N'|=|O'|$ and all objects and agents are acceptable as follows.
\begin{equation}\label{N'}N'=N\cup D\end{equation}
where $D=\{d_1,\ldots, d_m\}$ is a set of dummy agents (in lexicographic order $d_1,\ldots, d_m$).
\begin{equation}\label{O'}O'=O\cup \Phi\end{equation}
where $\Phi=\{\phi_1,\ldots, \phi_n\}$ is a set of null objects  (in lexicographic order $\phi_1,\ldots, \phi_n$).\medskip

Note that $|N'|=|O'|=n+m$.\medskip

Next, we extend preferences / priorities from the general instance $I$ to the associated instance $I'$ as follows. For any subset $B$ of $N$ or $O$, we denote the restriction of preference / priority of agent / object $\alpha$, $\succsim_\alpha$, to set $B$ by $\succsim_{\alpha}\!\!(B)$. Furthermore, for any subset $B$ of set $N$, $D$, $O$, or $\Phi$, we denote the lexicographic order of $B$ by $lex(B)$. The main idea of extending preferences and priorities from general instance $I$ to associated instance $I'$ is that each agent $i\in N$ has a default personal null object $\phi_i$ that is less preferred than all acceptable objects, more preferred than all the unacceptable objects, and that ranks agent $i$ as its highest priority agent. Furthermore, each object $o_j\in O$ has a default personal dummy agent $d_j$ who has lower priority than all the acceptable agents, higher priority than all unacceptable agents, and who most prefers object~$o_j$.\medskip

\noindent Each agent $i\in N$ extends his preferences $\succsim_i$ by replacing the null object $\emptyset$ by the set of null objects $\Phi$ such that agent $i$'s null object $\phi_i$ is more preferred than all other null objects (in strict lexicographic order) - each comma below indicates strict preferences at $\succsim'_i$:
\begin{equation}\label{agentspref}
  \succsim'_i= \succsim_i\!\left(\{o\in O\midd o \succ_i \emptyset\}\right),\ \phi_i,\ lex\left(\Phi\setminus \{\phi_i\}\right),\ \succsim_i\!\left(\{o\in O\midd \emptyset\succ_i o\right)\}).
\end{equation}
Each object $o_j\in O$ extends its priorities by replacing the null object $\emptyset$ by the set of dummy agents $D$ such that object $o_j$'s dummy agent $d_j$ is more preferred than all other dummy agents (in strict lexicographic order)  - each comma below indicates strict priorities at $\succsim'_{o_j}$:
\begin{equation}\label{objectpri}
\succsim'_{o_j}=\succsim_{o_j}\!\!\left(\{i\in N\midd i\succ_{o_j}\emptyset\}\right),\ d_j,\ lex\left(D\setminus \{d_j\}\right),\ \succsim_{o_j}\!\!\left(\{i\in N\midd \emptyset\succ_{o_j}i\}\right).
\end{equation}
Each dummy agent $d_j$ first prefers object $o_j$, then all other objects in $O\setminus\{o_j\}$ (in strict lexicographic order), and finally dummy agent $d_j$ ranks the null objects in $\Phi$ exactly as object $o_j$ ranks the agents in $N$, i.e., for all $i,k\in N$, $\phi_i \succsim_{d_j}' \phi_k$ if and only if $i \succsim_{o_j} k$  - each comma below indicates strict preferences at $\succsim'_{d_j}$:
\begin{equation}\label{dummypref}
\succsim'_{d_j}= o_j,\ lex\left(O\setminus \{o_j\}\right),\ \succsim'_{d_j}\!(\Phi).
\end{equation}
Each null object $\phi_i$ first ranks agent $i$ as the highest priority agent, then all the other agents in $N\setminus\{i\}$ (in strict lexicographic order), and finally null object $\phi_i$ ranks the dummy agents in $D$ exactly as agent $i$ ranks the objects in $O$, i.e., for all $o_j,o_k\in O$, $d_j \succsim_{\phi_i}' d_k$ if and only if $o_j \succsim_{i} o_k$  - each comma below indicates strict priorities at $\succsim'_{\phi_i}$:
\begin{equation}\label{nullpri}
\succsim'_{\phi_i}= i,\ lex(N\setminus \{i\}),\ \succsim'_{\phi_i}\!(D).
\end{equation}
Consider a generalized random matching $p$ for a general instance $I=(N,O,\succsim)$. Then, we defined the \emph{\textbf{associated random matching}} $p'$ for instance $I'=(N',O',\succsim')$ as follows:
\begin{itemize}
  \item for each pair $(i,o_j)\in N\times O$, $p'(i,o_j)=p(i,o_j)$ ,
  \item for each pair $(d_j,\phi_i)\in D\times \Phi$, $p'(d_j,\phi_i)=p(i,o_j)$,
  \item for each pair $(i,\phi_i)\in N\times \Phi$, $p'(i,\phi_i)=1-\sum_{o\in O}p(i,o)$,
  \item for each pair $(d_j,o_j)\in D\times O$, $p'(d_j,o_j)=1-\sum_{i\in N}p(i,o_j)$, and
  \item for all remaining pairs $(a,b)\in \left(N\cup D\right)\times\left(O\cup \Phi\right)$, $p'(a,b)=0$.
\end{itemize}

The associated random matching $p'$ looks as follows:
\begin{center}
\begin{blockarray}{ccccccccccc}
    &&\matindex{$o_1$} & \matindex{$\ldots$} & \matindex{$o_m$}&&\matindex{$\phi_1$}&\matindex{$\ldots$}&\matindex{$\phi_n$}\\
    \begin{block}{c(cccccccccc)}
   \matindex{$1$}& &$p(1,o_1)$   &$\ldots$  & $p(1,o_m)$&$|$&$p(1,\emptyset)$&&\Large{$0$}&\\
   \matindex{$\vdots$}& & $\vdots$  &   &$\vdots$&$|$&&$\ddots$&\\
   \matindex{$n$} & & $p(n,o_1)$& $\ldots$  & $p(n,o_m)$&$|$&\Large{$0$}&&$p(n,\emptyset)$\\
       &&---  & ---&--- &---&---&---&---&\\
   \matindex{$d_1$} &&$p(\emptyset,o_1)$ &  & \Large{$0$}&$|$&$p(1,o_1)$&$\ldots$&$p(n,o_1)$&\\
       \matindex{$\vdots$}& &  &$\ddots$ & &$|$&$\vdots$&&$\vdots$&\\
         \matindex{$d_m$}& &\Large{$0$} &  & $p(\emptyset,o_m)$&$|$&$p(1,o_m)$&$\ldots$&$p(n,o_m)$&\\
    \end{block}
  \end{blockarray}
  \end{center}
where
\begin{itemize}
	\item for each $i\in\{1,\ldots,n\}$, $p(i,\emptyset):=1-\sum_{o\in O}p(i,o)$ and
	\item for each $j\in\{1,\ldots,m\}$, $p(\emptyset,o_j):=1-\sum_{i\in N}p(i,o_j)$.
\end{itemize}

\begin{definition}[\textbf{Associated random matching $\bm{p'}$ respecting non-wastefulness}]\label{def:p'respectingNW}
\normalfont	Let $p$ be a generalized random matching and $p'$ its associated random matching. Then, $p'$ respects non-wastefulness (of $p$) if $p$ is non-wasteful; i.e., there is no acceptable pair $(i,o_j)\in N\times O$ such that [$\sum_{o':o'\succsim_i o_j} p'(i,o')<1$ and $\sum_{j\in N}p'(j,o_j)<1$].
\end{definition}

\begin{definition}[\textbf{Associated random matching $\bm{p'}$ respecting individual rationality}]\label{def:p'respectingIR}
\normalfont	Let $p$ be a generalized random matching and $p'$ its associated random matching. Then, $p'$ respects individual rationality (of $p$) if $p$ is individually rational; i.e., for any unacceptable pair $(i,o_j)\in N\times O$, $p'(i,o_j)=p(i,o_j)=0$ and $p'(d_j,\phi_i)=p(i,o_j)=0$.
\end{definition}

\begin{example}[\textbf{Transforming a general instance and a generalized random matching}]\normalfont Consider the following general instance $I=(N,O,\succ)$ with strict preferences and strict priorities: $N=\{1,2,3\}$, $O=\{x,y\}$,
\begin{center}
\begin{tabular}{lccc}
$\succ_1$:&$x$&$y$&$\emptyset$\\
$\succ_2$:&$y$&$x$&$\emptyset$\\
$\succ_3$:&$x$&$\emptyset$&$y$
\end{tabular}
\quad\quad
\begin{tabular}{lcccc}
$\succ_x$:&$2$&$3$&$1$&$\emptyset$\\
$\succ_y$:&$1$&$2$&$\emptyset$&$3$
\end{tabular}	
\end{center}	
and the generalized random matching
\begin{center}
\begin{tabular}{c}
$p=\begin{pmatrix}
\nicefrac{1}{3}&\nicefrac{1}{2}\\
0&\nicefrac{1}{2}\\
\nicefrac{2}{3}&0
\end{pmatrix}.$
\end{tabular}										
\end{center}
The associated instance equals $I'=(N',O', \succ')$ with strict preferences and strict priorities such  that $N'=\{1,2,3,d_x,d_y\}$, $O'=\{a,b,\phi_1,\phi_2,\phi_3\}$, and
\begin{center}
\begin{tabular}{lccccc}
$\succ'_1$:&$x$&$y$&$\phi_1$&$\phi_2$&$\phi_3$\\
$\succ'_2$:&$y$&$x$&$\phi_2$&$\phi_1$&$\phi_3$\\
$\succ'_3$:&$x$&$\phi_3$&$\phi_1$&$\phi_2$&$y$\\
$\succ'_{d_x}$:&$x$&$y$&$\phi_2$&$\phi_3$&$\phi_1$\\
$\succ'_{d_y}$:&$y$&$x$&$\phi_1$&$\phi_2$&$\phi_3$
\end{tabular}
\quad\quad
\begin{tabular}{lccccc}
$\succ'_x$:&$2$&$3$&$1$&$d_x$&$d_y$\\
$\succ'_y$:&$1$&$2$&$d_y$&$d_x$&$3$\\
$\succ'_{\phi_1}$:&$1$&$2$&$3$&$d_x$&$d_y$\\
$\succ'_{\phi_2}$:&$2$&$1$&$3$&$d_y$&$d_x$\\
$\succ'_{\phi_3}$:&$3$&$1$&$2$&$d_x$&$d_y$.
\end{tabular}	
\end{center}
The associated random matching equals
\begin{center}
\begin{tabular}{c}
$p'=$\begin{blockarray}{cccccccccc}
    &&\matindex{$x$} & \matindex{$y$} & & \matindex{$\phi_1$}&\matindex{$\phi_2$}&\matindex{$\phi_3$}\\
    \begin{block}{c(ccccccccc)}
   \matindex{$1$}& &$\nicefrac{1}{3}$   &$\nicefrac{1}{2}$  &$|$ & $\nicefrac{1}{6}$&$0$&$0$\\
   \matindex{$2$}& & $0$  & $\nicefrac{1}{2}$ &$|$ &$0$&$\nicefrac{1}{2}$&$0$\\
   \matindex{$3$} & & $\nicefrac{2}{3}$& $0$  & $|$ & $0$&$0$& $\nicefrac{1}{3}$\\
       &&---  & ---&--- &---&---&---&\\
   \matindex{$d_x$} && $0$& $0$ & $|$&$\nicefrac{1}{3}$&$0$&$\nicefrac{2}{3}$&\\
   \matindex{$d_y$}& & $0$ & $0$ & $|$&$\nicefrac{1}{2}$&$\nicefrac{1}{2}$&$0$\\
    \end{block}
  \end{blockarray}\ .
\end{tabular}										
\end{center}\hfill~$\diamond$
\end{example}

In the sequel, whenever preferences and priorities are strict, then weak stability (Definition~\ref{def:weakstability}) can be simply referred to as stability. The reason why we need to be more careful in terminology when preferences and priorities are weak is the fact that we use strict blocking when defining weak stability while other blocking notions are theoretically also possible (see also Section~\ref{subsection:weakprefpri}).\medskip

Next, we show that for each generalized deterministic matching $p$, no envy, individual rationality, and non-wastefulness are equivalent to weak stability of the associated deterministic matching $p'$.

\begin{proposition}\label{Prop:weakNJE-IR-NWiffStable}
The generalized deterministic matching $p$ has no envy and is individually rational and non-wasteful if and only if the associated deterministic matching $p'$ is weakly stable.
\end{proposition}
\begin{proof}[\textbf{Proof}]
\textbf{Part 1:} Let $p$ be a generalized deterministic matching and $p'$ its associated deterministic matching. We first show that $p$ being individually irrational, or wasteful, or having envy implies that $p'$ is not weakly stable. We do so via the following table that for any of the possible violations for $p$ lists an associated no-envy violation for $p'$. In Table~\ref{table:weakviolations1}, $(i,o_j)\in N\times O$:
\begin{center}
\begin{longtable}{ll}
	\captionsetup{width=.7\textwidth}
	\caption{Violations for $p$ and associated no-envy violation for $p'$. \textbf{IIR} $=$ individually irrational, \textbf{W} $=$ wasteful, \textbf{E} $=$ envy.}	\label{table:weakviolations1}\smallskip\\
  \hline
  \textbf{violation for} $\bm{p}$ & \textbf{associated no-envy violation for} $\bm{p'}$\\\hline\hline
  \textbf{IIR} agent $i$'s match is an unacceptable object:& \textbf{E} agent $i$ envies $d_j$ for $\phi_i$:\\
  $p(i,o_j)=1$ and $\emptyset\succ_i o_j$ & $p'(i,o_j)=1$, $p'(d_j,\phi_i)=1$,\\
  & $\phi_i\succ'_i o_j$, $i\succ'_{\phi_i}d_j$\\\hline
  \textbf{IIR} object $o_j$'s match is an unacceptable agent:& \textbf{E} agent $d_j$ envies $i$ for $o_j$: \\
  $p(i,o_j)=1$ and $\emptyset\succ_{o_j} i$ & $p'(i,o_j)=1$, $p'(d_j,\phi_i)=1$,\\
  & $o_j\succ'_{d_j} \phi_i$, $d_j\succ'_{o_j}i$\\\hline
  \textbf{W} agent $i$ gets $\emptyset$ and wants unassigned object $o_j$: & \textbf{E} agent $i$ envies $d_j$ for $o_j$:\\
  $i\succ_{o_j}\emptyset$, $o_j\succ_i \emptyset$,  & $p'(i,\phi_i)=1$, $p'(d_j,o_j)=1$,\\
  $\sum_{l\in N}p(l,o)=0$, $\sum_{o'\in O}p(i,o')=0$ & $o_j\succ'_i\phi_i$, $i\succ'_{o_j}d_j$\\\hline
  \textbf{W} agent $i$ gets $o'$ and wants unassigned object $o_j$: & \textbf{E} agent $i$ envies $d_j$ for $o_j$:\\
  $i\succ_{o_j}\emptyset$, $o_j\succ_i \emptyset$, & $p'(i,o')=1$, $p'(d_j,o_j)=1$,\\
  $\sum_{l\in N}p(l,o)=0$, $p(i,o')=1$ &  $o_j\succ'_i o'$, $i\succ'_{o_j}d_j$\\\hline
  \textbf{E} agent $i$ gets $o'$ and envies $k$ for $o_j$: & \textbf{E} agent $i$ envies $k$ for $o_j$:\\
  $p(i,o')=1$, $p(k,o_j)=1$,  &   $p'(i,o')=1$, $p'(k,o_j)=1$,\\ $o_j\succ_i o'$, $i\succ_{o_j}k$ & $o_j\succ'_i o'$, $i\succ'_{o_j}k$\\\hline
  \textbf{E} agent $i$ gets $\emptyset$ and envies $k$ for $o_j$: & \textbf{E} agent $i$ envies $k$ for $o_j$:\\
  $\sum_{o'\in O}p(i,o')=0$, $p(k,o_j)=1$,  &   $p'(i,\phi_i)=1$, $p'(k,o_j)=1$,\\ $o_j\succ_i \emptyset$, $i\succ_{o_j}k$ & $o_j\succ'_i \phi_i$, $i\succ'_{o_j}k$\\\hline
\end{longtable}
\end{center}
\textbf{Part 2:} We show that $p'$ not being weakly stable implies that $p$ is individually irrational, or wasteful, or has envy, or that the weak stability violation at $p'$ was not possible. Assume that at $p'$ some agent $a$ envies an agent $b$ for object $c$. Then, $p'(a,d)=1$, $p'(b,c)$, $c\succ'_{a}d$, and $a\succ'_{c}b$. Depending on the specifications of $a$, $b$, $c$, and $d$, different violations can be identified for $p$. The following tables list all no-envy violations for $p'$ and for $p$ associates individual irrationality, wastefulness, or envy (Table~\ref{table:weakviolations2}) or explains why the no-envy violations of $p'$ cannot occur given its definition (Table~\ref{table:weakviolations3}). Note that since $(a,d),(b,c)\in \{N\times O, N\times \Phi, D\times O, D\times \Phi\}$ we have in total 16 different cases to discuss.
\begin{center}
\begin{longtable}{ll}\captionsetup{width=.7\textwidth}
	\caption{No-envy violations for $p'$ and associated violation for $p$. \textbf{IIR} $=$ individually irrational, \textbf{W} $=$ wasteful, \textbf{E} $=$ envy.}	\label{table:weakviolations2}\smallskip\\
  \hline
\textbf{no-envy violation for} $\bm{p'}$ & \textbf{associated violation for} $\bm{p}$\\\hline\hline
\textbf{E} agent $a$ envies agent $b$ for $c$: & \textbf{E} agent $i$ envies agent $j$ for $o$:\\
$(a,d),(b,c)\in N\times O$,  &  \\
$(a,d)=(i,o')$, $(b,c)=(j,o)$, &  \\
$p'(i,o')=1$, $p'(j,o)=1$, & $p(i,o')=1$, $p(j,o)=1$,\\
$o\succ'_{i}o'$, $i\succ'_o j$. & $o\succ_{i}o'$, $i\succ_o j$.\\\hline

\textbf{E} agent $a$ envies agent $b$ for $c$: & agent $i$ is individually irrational\\
$(a,d)\in N\times O$, $(b,c)\in D\times O$,  &  or object $o$ is wasted:\\
$(a,d)=(i,o')$, $(b,c)=(d_j,o_j)$, &   \\
$p'(i,o')=1$, $p'(d_j,o_j)=1$, &  $p(i,o')=1$, $\sum_{k\in N}p(k,o_j)=0$,\\
$o_j\succ'_i o'$, $i\succ'_{o_j}d_j$, & $o_j\succ_{i}o'$, $i\succ_{o_j} \emptyset$,\\
(a) $\phi_i\succ'_i o_j$, & (a) \textbf{IIR} $\emptyset\succ_i o_j$, \\
(b) $o_j\succ'_i\phi_i$. & (b) \textbf{W} $o_j\succ_i\emptyset$.\\\hline

\textbf{E} agent $a$ envies agent $b$ for $c$: & \textbf{IIR} agent $i$ is individually irrational:\\
$(a,d)\in N\times O$, $(b,c)\in D\times \Phi$,  &  \\
$(a,d)=(i,o')$, $(b,c)=(d_j,\phi_k)$, &  \\
in particular, $p'(i,o')=1$, $\phi_k\succ'_i o'$. &  $p(i,o')=1$, $\emptyset\succ_i o'$.\\\hline

\textbf{E} agent $a$ envies agent $b$ for $c$: & \textbf{W} object $o_j$ is wasted: \\
$(a,d)\in N\times \Phi$, $(b,c)\in D\times O$,  &  \\
$(a,d)=(i,\phi_i)$, $(b,c)=(d_j,o_j)$, &     \\
$p'(i,\phi_i)=1$, $p'(d_j,o_j)=1$, &  $\sum_{o'\in O}p(i,o')=0$, $\sum_{k\in N}p(k,o_j)=0$,\\
$o_j\succ'_i\phi_i$, $i\succ'_{o_j}d_j$. & $o_j\succ_{i}\emptyset$, $i\succ_{o_j} \emptyset$.\\\hline

\textbf{E} agent $a$ envies agent $b$ for $c$: & \textbf{E} agent $l$ envies agent $k$ for $o_i$:\\
$(a,d),(b,c)\in D\times \Phi$,  & \\
$(a,d)=(d_i,\phi_k)$, $(b,c)=(d_j,\phi_l)$, &  \\
$p'(d_i,\phi_k)=1$, $p'(d_j,\phi_l)=1$, & $p(k,o_i)=1$, $p(l,o_j)=1$,\\
$\phi_l\succ'_{d_i}\phi_k$, $d_i\succ'_{\phi_l} d_j$. &  $l\succ_{o_i}k$, $o_i\succ_l o_j$.\\\hline

\textbf{E} agent $a$ envies agent $b$ for $c$:& \textbf{JE} agent $i$ envies agent $j$ for $o$:\\
$(a,d)\in N\times \Phi$, $(b,c)\in N\times O$,  &  \\
$(a,d)=(i,\phi_i)$, $(b,c)=(j,o)$, &  \\
$p'(i,\phi_i)=1$, $p'(j,o)=1$, & $\sum_{o':o'\succsim o}p(i,o')=0$, $p(j,o)=1$,\\
$o\succ'_i\phi_i$, $i\succ'_o j$. & $o\succ_i \emptyset$, $i\succ_o j$.\\\hline

\textbf{E} agent $a$ envies agent $b$ for $c$: & \textbf{IIR} object $o$ is individually irrational:\\
$(a,d)\in D\times O$, $(b,c)\in N\times O$,  &  \\
$(a,d)=(d_i,o_i)$, $(b,c)=(j,o)$, &  \\
in particular, $p'(j,o)=1$, $d_i\succ'_o j$. &  $p(j,o)=1$, $\emptyset\succ_o j$.\\\hline

\textbf{E} agent $a$ envies agent $b$ for $c$: & \textbf{IIR} object $o$ is individually irrational:\\
$(a,d)\in D\times \Phi$, $(b,c)\in N\times O$,  &  \\
$(a,d)=(d_i,\phi_k)$, $(b,c)=(j,o)$, &  \\
in particular, $p'(j,o)=1$, $d_i\succ'_o j$. &  $p(j,o)=1$, $\emptyset\succ_o j$.\\\hline
\end{longtable}
\end{center}

\smallskip

\begin{center}
\begin{longtable}{ll}\captionsetup{width=.6\textwidth}
	\caption{No-envy violations of $p'$ that are not possible. \textbf{E} $=$ envy, \textbf{IP} $=$ impossibility.}
	\label{table:weakviolations3}\smallskip\\
  \hline
\textbf{no-envy violation for} $\bm{p'}$ & \textbf{why this no-envy violation for} $\bm{p'}$ \textbf{is not possible}\\\hline\hline
\textbf{E} agent $a$ envies agent $b$ for $c$: & \\
$(a,d)\in N\times O$, $(b,c)\in N\times \Phi$,  & \textbf{IP} $i\succ'_{\phi_j} j$ is not possible \\
$(a,d)=(i,o')$, $(b,c)=(j,\phi_j)$, &  because by  the definition of $p'$, \\
in particular, $i\succ'_{\phi_j} j$. & $j\succ'_{\phi_j}i$.
\\\hline

\textbf{E} agent $a$ envies agent $b$ for $c$: & \\
$(a,d)\in N\times \Phi$, $(b,c)\in N\times \Phi$,  & \textbf{IP} $\phi_j\succ'_{i} \phi_i$ is not possible\\
$(a,d)=(i,\phi_i)$, $(b,c)=(j,\phi_j)$,& because by the definition of $p'$,\\
in particular, $\phi_j\succ'_{i} \phi_i$. & $\phi_i\succ'_{i}\phi_j$.\\\hline

\textbf{E} agent $a$ envies agent $b$ for $c$: & \\
$(a,d)\in N\times \Phi$, $(b,c)\in D\times \Phi$,  & \textbf{IP} $\phi_k\succ'_{i} \phi_i$ is not possible\\
$(a,d)=(i,\phi_i)$, $(b,c)=(d_j,\phi_k)$, & because by the definition of $p'$, \\
in particular, $\phi_k\succ'_{i} \phi_i$. & $\phi_i\succ'_{i}\phi_k$.\\\hline

\textbf{E} agent $a$ envies agent $b$ for $c$: & \\
$(a,d)\in D\times O$, $(b,c)\in D\times O$,  & \textbf{IP} $o_j\succ'_{d_i} o_i$ is not possible\\
$(a,d)=(d_i,o_i)$, $(b,c)=(d_j,o_j)$, & because by the definition of $p'$, \\
in particular, $o_j\succ'_{d_i} o_i$. & $o_i\succ'_{d_i} o_j$.\\\hline

\textbf{E} agent $a$ envies agent $b$ for $c$: & \\
$(a,d)\in D\times O$, $(b,c)\in D\times \Phi$,  & \textbf{IP} $\phi_k\succ'_{d_i} o_i$ is not possible\\
$(a,d)=(d_i,o_i)$, $(b,c)=(d_j,\phi_k)$, & because by the definition of $p'$, \\
in particular, $\phi_k\succ'_{d_i} o_i$. & $o_i\succ'_{d_i} \phi_k$.\\\hline

\textbf{E} agent $a$ envies agent $b$ for $c$: & \\
$(a,d)\in D\times O$, $(b,c)\in N\times \Phi$,  & \textbf{IP} $\phi_j\succ'_{d_i} o_i$ is not possible\\
$(a,d)=(d_i,o_i)$, $(b,c)=(j,\phi_j)$, & because by the definition of $p'$, \\
in particular, $\phi_j\succ'_{d_i} o_i$. & $o_i\succ'_{d_i} \phi_j$.\\\hline

\textbf{E} agent $a$ envies agent $b$ for $c$: & \\
$(a,d)\in D\times \Phi$, $(b,c)\in N\times \Phi$,  & \textbf{IP} $d_i\succ'_{\phi_j}j$ is not possible\\
$(a,d)=(d_i,\phi_k)$, $(b,c)=(j,\phi_j)$, & because by the definition of $p'$, \\
in particular, $d_i\succ'_{\phi_j}j$. & $j\succ'_{\phi_j}d_i$.\\\hline

\textbf{E} agent $a$ envies agent $b$ for $c$: & \\
$(a,d)\in D\times \Phi$, $(b,c)\in D\times O$,  & \textbf{IP} $d_i\succ'_{o_j}d_j$ is not possible\\
$(a,d)=(d_i,\phi_k)$, $(b,c)=(d_j,o_j)$, & because by the definition of $p'$, \\
in particular, $d_i\succ'_{o_j}d_j$. & $d_j\succ'_{o_j}d_i$.\\\hline
\end{longtable}\qedhere
\end{center}
\end{proof}

\begin{example}[\textbf{No-envy, individual rationality, and non-wastefulness are logically independent}]
\normalfont Consider the following general instance $I=(N,O,\succ)$ with strict preferences and strict priorities: $N=\{1,2\}$, $O=\{x,y\}$,
\begin{center}
\begin{tabular}{lccc}
$\succ_1$:&$x$&$\emptyset$&$y$\\
$\succ_2$:&$x$&$y$&$\emptyset$
\end{tabular}
\quad\quad
\begin{tabular}{lccc}
$\succ_x$:&$2$&$1$&$\emptyset$\\
$\succ_y$:&$1$&$2$&$\emptyset$
\end{tabular}	
\end{center}	
and the generalized deterministic matchings
\begin{center}
\begin{tabular}{c}
$p^1=\begin{pmatrix}
0&0\\
0&0
\end{pmatrix}$,
\end{tabular}
\quad\quad	
\begin{tabular}{c}
$p^2=\begin{pmatrix}
0&1\\
1&0
\end{pmatrix}$,
\end{tabular}
\quad\quad	
\begin{tabular}{c}
$p^3=\begin{pmatrix}
1&0\\
0&1
\end{pmatrix}$,
\end{tabular}
\quad\quad	
\begin{tabular}{c}
$p^4=\begin{pmatrix}
0&0\\
1&0
\end{pmatrix}$.
\end{tabular}										
\end{center}
Then, $p^1$ has no envy, is individually rational, but it is wasteful; $p^2$ has no envy, is non-wasteful, but it is individually irrational; $p^3$ is individually rational, non-wasteful, but it has envy; $p^4$ is the only generalized deterministic matching for this instance that has no envy, is individually rational, and is non-wasteful.\hfill~$\diamond$
\end{example}

The classic definition of weak stability for generalized deterministic matchings in our model is the following.

\begin{definition}[\textbf{Weak stability for generalized deterministic matchings}]
\normalfont	A generalized deterministic matching $p$ is \emph{\textbf{weakly stable}} if it is individually rational and there exist no agent and no object that would prefer each other to their current match, i.e., there exists no pair $(i,o)\in N\times O$ such that $\sum_{o':o'\succsim_i o}p(i,o')=0$ (agent $i$ would like to have $o$) and $\sum_{j:j\succsim_o i}p(j,o)=0$  (object $o$ would like to be matched to $i$).\label{def:weakgeneralstable}
\end{definition}

The well-known deferred-acceptance algorithm \citep{GS62} computes a generalized deterministic matching that is weakly stable.\medskip

Whenever preferences and priorities are strict, then weak stability (Definition~\ref{def:weakgeneralstable})  can be simply referred to as stability. It is easy to check that the following now holds.

\begin{proposition}A generalized deterministic matching is weakly stable if and only if it is individually rational, non-wasteful, and has no envy.\label{Prop:weakStableWiffNJE-IR-NW}
\end{proposition}
\begin{proof}[\textbf{Proof}]
Let $p$ be a generalized deterministic matching that is individually rational. Assume $p$ is weakly stable, i.e., there exists no pair $(i,o)\in N\times O$ such that $\sum_{o':o'\succsim_i o}p(i,o')=0$ and $\sum_{j:j\succsim_o i}p(j,o)=0$. Since $p$ is deterministic, this is equivalent to there being no pair $(i,o)\in N\times O$ such that $\sum_{o':o'\succsim_i o}p(i,o')=0$ and (a) $\sum_{j\in N}p(j,o)=0$ or (b) for some agent $j\in N$, $i\succ_o j$ and $p(j,o)=1$. This in turn is equivalent to $p$ being (a) non-wasteful and (b) having no envy.
\end{proof}

Propositions~\ref{Prop:weakNJE-IR-NWiffStable} and \ref{Prop:weakStableWiffNJE-IR-NW} now imply the following (see Figure~\ref{fig:weakassociated-deterministic}).

\begin{proposition}\label{Prop:weakStableiffStable}
A generalized deterministic matching $p$ is weakly stable if and only if the associated deterministic matching $p'$ is weakly stable.
\end{proposition}

\begin{figure}[H]%[!htbp]
\begin{framed}
\textbf{Propositions~\ref{Prop:weakNJE-IR-NWiffStable}, \ref{Prop:weakStableWiffNJE-IR-NW}, and \ref{Prop:weakStableiffStable}:} For any generalized  deterministic matching $p$ and its associated deterministic matching $p'$, we have\bigskip
\begin{center}
\scalebox{1}{
\begin{tikzpicture}
\tikzstyle{pfeil}=[->,>=angle 60, shorten >=1pt,draw]
\tikzstyle{onlytext}=[]
    \node[onlytext] (p) at (0,0) {\begin{tabular}{cl}
 & no envy,\\     $p$ &  non-wastefulness,\\
 &individual rationality\\
 &(Defs.~\ref{def:weakgeneralnoenvy}, \ref{def:weakgeneralnonwastefulness}, and \ref{def:weakgeneralindividualrationality})
\end{tabular}
};

    \node[onlytext] (p') at (3,-4) {\begin{tabular}{cl}
 & no envy\\    $p'$  & $=$ weakly stable\\
 &(Def.~\ref{def:weakstability})
\end{tabular}
};

    \node[onlytext] (p-stable) at (-3,-4) {\begin{tabular}{cl}
$p$ & weakly stable\\
 &(Def.~\ref{def:weakgeneralstable})
\end{tabular}
};

\draw[pfeil, thick, blue] (p) to (p');
\draw[pfeil, thick, blue] (p') to (p);
 \node[onlytext] (pp') at (2.5,-2) {Prop.~\ref{Prop:weakNJE-IR-NWiffStable}};

\draw[pfeil, thick, blue] (p-stable) to (p');
\draw[pfeil, thick, blue] (p') to (p-stable);
 \node[onlytext] (pp') at (-0.3,-4.5) {Prop.~\ref{Prop:weakStableiffStable}};

\draw[pfeil, thick, blue] (p-stable) to (p);
\draw[pfeil, thick, blue] (p) to (p-stable);
 \node[onlytext] (pp-stable) at (-2.5,-2) {Prop.~\ref{Prop:weakStableWiffNJE-IR-NW}};
\end{tikzpicture}}
\end{center}	
\caption{Relations between weak stability, no envy, individual rationality, and non-wastefulness for generalized deterministic matchings.}
\label{fig:weakassociated-deterministic}
\end{framed}
\end{figure}

By Proposition~\ref{Prop:weakStableWiffNJE-IR-NW}, a generalized deterministic matching $p$ is weakly stable if it is individually rational, non-wasteful, and has no envy. Recall that no envy implies that there exist no $i,j\in N$ and no $o\in O$ such that $\sum_{o':o'\succsim_i o}p(i,o')=0$, $p(j,o)=1$, and $i\succ_{o}j$. The latter is equivalent to the following inequalities being satisfied:\footnote{For instances with strict preferences and strict priorities, this characterization of stable matchings is due to \citet{Rothblum1992} \citep[see also][]{RothetalMOR93}.} for each acceptable pair $(i,o)\in N\times O$,
\begin{equation}\label{weakgeneralLE4W}
p(i,o)+\sum_{o':o'\succsim_i o;o'\neq o}p(i,o')+ \sum_{j:j\succsim_{o}i;j\neq i}p(j,o)\geq 1.
\end{equation}

We now adapt all previous stability concepts introduced in Section~\ref{subsection:weakprefpri} to generalized random matchings. First, we adjust the property of no ex-ante envy to generalized random matchings.

\begin{definition}[\textbf{No ex-ante envy for generalized random matchings}]
\normalfont	A generalized random matching $p$ has \emph{\textbf{no ex-ante envy}} if there exists no agent $i$ who prefers a higher probability for object $o$ while object $o$ is matched with positive probability to some agent $j$ with lower priority than $i$, i.e., there exist no $i,j\in N$ and no $o\in O$ such that $\sum_{o':o'\succsim o}p(i,o')<1$ (agent $i$ would like to have more of $o$), $p(j,o)>0$ (agent $j$ has some of $o$), $o\succ_i o'$, and $i\succ_oj$.\label{def:weakgeneralnoexantejustifiedenvy}
\end{definition}

For generalized random matchings the definition of Aharoni-Fleiner fractional stability (Definition~\ref{def:AFfractionalstability}) remains the same and its equivalence to no ex-ante envy follows as before.\medskip

Next, for each generalized random matching $p$, no ex-ante envy, individual rationality, and non-wastefulness are equivalent to ex-ante weak stability of the associated random matching~$p'$.

\begin{proposition}\label{Prop:weakex-anteNJE-IR-NWiffex-anteStable}
The generalized random matching $p$ has no ex-ante envy and is individually rational and non-wasteful if and only if the associated random matching $p'$ is ex-ante weakly stable.
\end{proposition}
\begin{proof}[\textbf{Proof}]
The proof follows exactly along the lines of the proof of Proposition~\ref{Prop:weakNJE-IR-NWiffStable}. The only difference is that in that proof no envy, individual rationality, non-wastefulness, and weak stability all are defined for probabilities 1 and 0 to receive an object and when we now consider no ex-ante envy, individual rationality, non-wastefulness, and ex-ante weak stability, these definitions pertain to any probability of receiving an object: all arguments that were using an agent receiving an object with probability 1 now apply for an agent receiving a positive probability of that object.
\end{proof}

Ex-ante weak stability for generalized random matchings is naturally defined as follows.

\begin{definition}[\textbf{Ex-ante weak stability for generalized random matchings}]
\normalfont	A generalized random matching $p$ is \emph{\textbf{ex-ante weakly stable}} if it is individually rational and there exist no agent and no object that would prefer a higher probability for each other, i.e., there exist no pair $(i,o)\in N\times O$ such that $\sum_{o':o'\succsim_i o}p(i,o')<1$ (agent $i$ would like to have more of $o$) and $\sum_{j:j\succsim_o i}p(j,o)<1$  (object $o$ would like to be matched more to $i$).\label{def:weakgeneralexantestable}
\end{definition}

It is easy to check that the following now holds.

\begin{proposition}A generalized random matching is ex-ante weakly stable if and only if it has no ex-ante envy and it is individually rational and non-wasteful.\label{Prop:weakex-anteStableWiffNJE-IR-NW}
\end{proposition}
\begin{proof}[\textbf{Proof}]
Let $p$ be a generalized random matching that is individually rational. Assume $p$ is ex-ante weakly stable, i.e., there exists no pair $(i,o)\in N\times O$ such that $\sum_{o':o'\succsim_i o}p(i,o')<1$ and $\sum_{j:j\succsim_o i}p(j,o)<1$. This is equivalent to there being no pair $(i,o)\in N\times O$ such that $\sum_{o':o'\succsim_i o}p(i,o')<1$ and (a) $\sum_{j\in N}p(j,o)<1$ or (b) for some agent $j\in N$, $i\succ_o j$ and $p(j,o)>0$. This in turn is equivalent to $p$ being (a) non-wasteful and (b) having no ex-ante envy.
\end{proof}

Propositions~\ref{Prop:weakex-anteNJE-IR-NWiffex-anteStable} and \ref{Prop:weakex-anteStableWiffNJE-IR-NW} now imply the following  (see the top part of Figure~\ref{fig:weakassociated-random}).

\begin{proposition}\label{Prop:weakex-anteStableiffex-anteStable}
The generalized random matching $p$ is ex-ante weakly stable if and only if the associated random matching $p'$ is ex-ante weakly stable.
\end{proposition}

Next, we adjust the properties of ex-post weak stability and robust ex-post weak stability to generalized random matchings.\medskip

Recall that each generalized random matching can be represented as a convex combination of generalized deterministic matchings. We now establish three results concerning the decomposition of an individually rational, respectively non-wasteful, generalized random matching.

\begin{lemma}
A generalized random matching is individually rational if and only if in each of its decompositions all generalized deterministic matchings are individually rational.
\end{lemma}
\begin{proof}[\textbf{Proof}]
\textbf{Part 1:} Suppose that generalized random matching $p$ is individually irrational. Then, for some $(i,o)\in N\times O$, $p(i,o)>0$ and agent $i$ or object $o$ considers the other unacceptable. Then, in any decomposition of $p$ into generalized deterministic matchings, there exists a generalized deterministic matching $q$ such that $q(i,o)=1$ and $q$ is individually irrational.\medskip

\noindent \textbf{Part 2:} Suppose that at some decomposition of $p$ there exists an individually irrational generalized deterministic matching $q$, i.e., for some $(i,o)\in N\times O$, $q(i,o)=1$ and agent $i$ or object $o$ considers the other unacceptable. Then,  $p(i,o)>0$  and $p$ is individually irrational.
\end{proof}

\begin{lemma}\label{lemma:non-wasteful-decomp}
If a generalized random matching is non-wasteful, then in each of its decompositions all generalized deterministic matchings are non-wasteful.
\end{lemma}
\begin{proof}
[\textbf{Proof}] Suppose that at some decomposition of $p$ there exists a wasteful generalized deterministic matching $q$, i.e., there exists an acceptable pair $(i,o)\in N\times O$ such that $\sum_{o':o'\succsim_i o}q(i,o')=0$ ($i$ would like to have object $o$) and $\sum_{j\in N}q(j,o)=0$ (object $o$ is not allocated). Then it follows that $\sum_{o':o'\succsim_i o}p(i,o')<1$ and $\sum_{j\in N}p(j,o)<1$. Hence, $p$ is wasteful.
\end{proof}

The following example shows that the converse statement in Lemma~\ref{lemma:non-wasteful-decomp} does not hold.

\begin{example}[\textbf{A wasteful generalized random matching that can be decomposed into generalized deterministic (non-wasteful) weakly stable  matchings}]\label{example:wastefulIR+inequalities}
\normalfont Consider the following general instance $I=(N,O,\succsim)$ with strict preferences and weak priorities (the brackets indicate indifferences): $N=\{1,2,3\}$, $O=\{x,y,z\}$,  	
\begin{center}
\begin{tabular}{lcccc}
$\succsim_1$:&$x$&$y$&$z$&$\emptyset$\\
$\succsim_2$:&$z$&$y$&$x$&$\emptyset$\\
$\succsim_3$:&$x$&$y$&$\emptyset$&$z$\\
\end{tabular}
\quad\quad
\begin{tabular}{lcccc}
$\succsim_{x}$:&$[1\ 3]$&$2$&$\emptyset$\\
$\succsim_{y}$:&$2$&$1$&$\emptyset$&$3$\\
$\succsim_{z}$:&$1$&$2$&$3$&$\emptyset$\\
\end{tabular}	
\end{center}
Consider the generalized random matching
\[p=\begin{pmatrix}
\nicefrac{1}{2}&0&\nicefrac{1}{2}\\
0&\nicefrac{1}{2}&\nicefrac{1}{2}\\
\nicefrac{1}{2}&0&0
\end{pmatrix}\]
and note that it is wasteful: agent 1 would like to have more of object $y$ that is not fully allocated. However, matching $p$ can be decomposed into two generalized deterministic non-wasteful and weakly stable matchings as follows:
\[p=\frac{1}{2}q^1+\frac{1}{2}q^2\]
where
\[q^1=\begin{pmatrix}
1&0&0\\
0&0&1\\
0&0&0
\end{pmatrix}\]							and
\[q^2=\begin{pmatrix}
0&0&1\\
0&1&0\\
1&0&0
\end{pmatrix}.\]
Note that $q^1$ and $q^2$ are non-wasteful and weakly stable: at $q^1$ both agents~1 and 2 are matched to their most preferred objects and the unassigned object $y$ finds agent~3 unacceptable; agent~3 cannot block with $x$ since $x$ has maximal priority for agent~1 and agent~3 cannot block with $z$ since $z$ has higher priority for agent~2. At $q^2$ objects $y$ and $z$ will not block because they are matched respectively to their highest priority agents; agent~1 would like to block with $x$ but $x$ has maximal priority for agent~3.\hfill~$\diamond$
\end{example}

Example~\ref{example:wastefulIR+inequalities} illustrates why in the next two definitions it is important to add non-wastefulness.

\begin{definition}[\textbf{Ex-post weak stability for generalized random matchings}]
\normalfont	A generalized random matching $p$ is \emph{\textbf{ex-post weakly stable}} if it is non-wasteful and can be decomposed into generalized deterministic weakly stable matchings.\label{def:weakgeneralexpostStability}
\end{definition}

\begin{definition}[\textbf{Robust ex-post weak stability for generalized random matchings}]
\normalfont	A generalized random matching $p$ is \emph{\textbf{robust ex-post weakly stable}} if it is non-wasteful and all of its decompositions are into generalized deterministic weakly stable matchings.\label{def:weakgeneralrobustexpoststability}
\end{definition}

We have the following equivalences for ex-post weak stability and robust ex-post weak stability for generalized random matchings and their associated random matchings.

\begin{proposition}\label{Prop:weakexpostStableiffexpostStable}
The generalized random matching $p$ is ex-post weakly stable if and only if the associated random matching $p'$ is ex-post weakly stable and respects non-wastefulness.\label{Prop:weakgeneralexpost-iff-expost}
\end{proposition}
\begin{proof}[\textbf{Proof}]
Let $p$ be a generalized random matching and $p'$ its associated random matching.\medskip

\noindent\textbf{Part 1:} Let $p$ be an ex-post weakly stable generalized matching. Recall that the non-wastefulness of $p$ is equivalent to $p'$ respecting non-wastefulness. Furthermore, $p$ can be decomposed into generalized deterministic weakly stable matchings. By Proposition~\ref{Prop:weakStableiffStable}, each generalized deterministic weakly stable matching in the decomposition corresponds to an associated deterministic weakly stable matching. The induced decomposition consisting of the associated deterministic weakly stable matchings is a decomposition of the associated random matching $p'$. Hence, $p'$ is ex-post weakly stable.\medskip

\noindent \textbf{Part 2:} Recall that from any associated random matching $p'$ we can obtain the original generalized random matching $p$ by taking its first $n$ rows and its first $m$ columns ($|N|=n$ and $|O|=m$). Let the associated random matching $p'$ of $p$ be ex-post weakly stable and respect non-wastefulness. Then, $p'$ can be decomposed into deterministic weakly stable matchings. Note that by taking the first $n$ rows and the first $m$ columns of each of the deterministic weakly stable matchings in the decomposition, we can derive a decomposition of $p$ into generalized deterministic weakly stable matchings (Proposition~\ref{Prop:weakStableiffStable}). Furthermore, since $p'$ respects non-wastefulness, $p$ is non-wasteful. Hence, $p$ is ex-post weakly stable.
\end{proof}

\begin{proposition}\label{Prop:weakrexpostStableiffrexpostStable}
The generalized random matching $p$ is robust ex-post weakly stable if and only if the associated random matching $p'$ is robust ex-post weakly stable and respects non-wastefulness.
\end{proposition}
\begin{proof}[\textbf{Proof}]Let $p$ be a generalized random matching and $p'$ its associated random matching. By Proposition~\ref{Prop:weakexpostStableiffexpostStable}, $p$ is ex-post weakly stable if and only if $p'$ is ex-post weakly stable and respects non-wastefulness.\medskip

\noindent\textbf{Part 1:} Let $p$ be an ex-post weakly stable generalized matching that is not robust ex-post weakly stable. Hence, $p$ has a decomposition into generalized deterministic matchings that is not weakly stable, i.e., at least one of the generalized deterministic matchings in the decomposition is not weakly stable. By Proposition~\ref{Prop:weakStableiffStable}, each generalized deterministic matching in the decomposition corresponds to an associated deterministic matching and the weakly unstable generalized deterministic matching leads to a weakly unstable associated deterministic matching. The induced decomposition consisting of the associated deterministic matchings is a decomposition of the associated random matching $p'$. Hence, $p'$ has a decomposition into deterministic matchings that are not all weakly stable and $p'$ is not robust ex-post weakly stable.\medskip

\noindent \textbf{Part 2:} Recall that from any associated random matching $p'$ we can obtain the original generalized random matching $p$ by taking its first $n$ rows and its first $m$ columns ($|N|=n$ and $|O|=m$). Let the associated random matching $p'$ of $p$ respect non-wastefulness and be ex-post weakly stable but not robust ex-post weakly stable.  Hence, $p'$ has a decomposition into deterministic matchings that is not weakly stable, i.e., at least one of the deterministic matchings in the decomposition is not weakly stable. Note that by taking the first $n$ rows and the first $m$ columns of each of the deterministic matchings in the decomposition, we can derive a decomposition of $p$ into generalized deterministic matchings and the weakly unstable associated deterministic matching leads to a weakly unstable generalized deterministic matching (Proposition~\ref{Prop:weakStableiffStable}). The induced decomposition consisting of the generalized deterministic matchings is a decomposition of the generalized random matching $p$. Hence, $p$ has a decomposition into generalized deterministic matchings that are not all weakly stable and $p$ is not robust ex-post weakly stable.
\end{proof}	

Next, we adjust the properties of fractional weak stability and claimwise weak stability to generalized random matchings.\medskip

Fractional weak stability is again obtained by relaxing the ``integer solution requirement'' for the inequalities that define weak stability for generalized deterministic matchings (\ref{weakgeneralLE4W}). Given a generalized random matching $p$ and an object $o$, recall that by $p(\emptyset,o)$ we denote the amount of object $o$ that is unassigned, i.e., $p(\emptyset,o)=1-\sum_{i\in N}p(i,o)$.

\begin{definition}[\textbf{Fractional weak stability and violations of fractional weak stability for generalized random matchings}]
\normalfont	A generalized random matching $p$ is \textit{\textbf{fractionally weakly stable}} if
$p$ is individually rational, non-wasteful, and for each acceptable pair $(i,o)\in N\times O$,
\begin{equation}\tag{\ref{weakgeneralLE4W}}
p(i,o)+\sum_{o':o'\succsim_i o;o'\neq o}p(i,o')+ \sum_{j:j\succsim_{o}i;j\neq i}p(j,o)\geq 1,
\end{equation}
or more compactly,
\begin{equation}\label{weakgeneralLE5W}
\sum_{o':o'\succsim_i o;o'\neq o}p(i,o')\geq \sum_{j:j\prec_o i}p(j,o)+p(\emptyset,o).
\end{equation}
A \textit{\textbf{violation of fractional weak stability}} occurs if there exists a pair $(i,o)\in N\times O$ such that
\begin{equation}\label{weakgeneralLE6}
\sum_{j:j\prec_o i}p(j,o)+p(\emptyset,o)>\sum_{o':o'\succsim_i o;o'\neq o}p(i,o').
\end{equation}
\label{def:weakgeneralfractionalstable}
\end{definition}
Inequality (\ref{weakgeneralLE6}) implies $\sum_{o':o'\succsim_i o}p(i,o')<1$, i.e., agent $i$ receives some fraction of an object in his strict lower contour set at $o$ or $i$ is not fully matched (if not, this would imply that $\sum_{o':o'\succsim_i o;o'\neq o}p(i,o')+p(i,o)=1$ and hence, $\sum_{j:j\prec_o i}p(j,o)+p(\emptyset,o)+p(i,o)>1$; a contradiction). Thus, agent $i$ would want to consume more of object $o$. Inequality (\ref{weakgeneralLE6}) also implies $\sum_{j:j\succsim_o i}p(j,o)<1$, i.e., object $o$ receives some fraction of an agent in its strict lower contour set at $i$ or $o$ is not fully allocated. Thus, object $o$ would want to consume more of agent $i$. Moreover, strict inequality (\ref{weakgeneralLE6}) encodes the following envy notion: using consumption process language, as long as agent $i$ consumes objects that are different and not worse than $o$ he does not envy the set of lower priority agents to jointly consume fractions of $o$ and he does not mind fractions of $o$ to be unassigned, however, once the unassigned amounts of $o$ plus the amounts lower priority agents have consumed reach agent $i$'s weak upper contour set at $o$ (not including $o$), agent $i$ starts either having envy or complaining about wastefulness (unless agent $i$ can fill his remaining probability quota with object $o$).

\begin{remark}[\textbf{A symmetric reformulation of fractional weak stability for generalized random matchings and its violations}]\label{remark:symmetricFractionalWeakgeneral}
\normalfont In the definition of fractional weak stability for generalized random matchings by inequalities (\ref{weakgeneralLE5W}) and of a violation of fractional weak stability for generalized random matchings by inequality (\ref{weakgeneralLE6}) we have taken the viewpoint of an agent who considers the consumptions of lower priority agents for an object and the amount of the object that is unassigned. The symmetric formulations when taking the viewpoint of an object that ``considers'' the matches of lower preferred objects to an agent and the amount of an agent he is not matched at are as follows. Given a generalized random matching $p$ and an agent $i$, recall that by $p(i,\emptyset)$ we denote the amount of agent $i$ that is not matched, i.e., $p(i,\emptyset)=1-\sum_{o\in O}p(i,o)$. Then, a generalized random matching $p$ is fractionally weakly stable if for each acceptable pair $(i,o)\in N\times O$,
\begin{equation}\tag{\ref{weakgeneralLE5W}'}
\sum_{j:j\succsim_o i;j\neq i}p(j,o)\geq \sum_{o':o'\prec_i o}p(i,o')+p(i,\emptyset).
\end{equation}
We can write a violation of fractional weak stability as, there exists an acceptable pair $(i,o)\in N\times O$ such that
\begin{equation}\tag{\ref{weakgeneralLE6}'}
\sum_{o':o'\prec_i o}p(i,o')+p(i,\emptyset)>\sum_{j:j\succsim_o i;j\neq i}p(j,o).
\end{equation}\hfill~$\diamond$
\end{remark}

Interestingly, if preferences and priorities are strict, then inequalities~(\ref{weakgeneralLE4W}) imply non-wastefulness.

\begin{proposition}\label{prop:ruralhospital-nonwastefulness}
Let $p$ be an individually rational generalized random matching such that for each acceptable pair $(i,o)\in N\times O$,
\[p(i,o)+\sum_{o':o'\succsim_i o;o'\neq o}p(i,o')+ \sum_{j:j\succsim_{o}i;j\neq i}p(j,o)\geq 1.\]
If preferences and priorities are strict, then $p$ satisfies non-wastefulness.
\end{proposition}
\begin{proof}[\textbf{Proof}]
\citet{RothetalMOR93} show that in the general model with strict preferences and priorities, any individually rational generalized random matching satisfying inequalities~(\ref{weakgeneralLE4W}) can be decomposed into non-wasteful and stable generalized deterministic matchings. On top of that,  the rural hospital theorem~\citep{Roth86a} implies that the set of matched agents and objects is always the same in all stable generalized deterministic matchings.  Now suppose, by contradiction, that a convex combination of non-wasteful and stable generalized deterministic matchings $q^1,\ldots, q^m$ leads to a wasteful generalized random matching $p$. By definition of wastefulness, there is an acceptable pair $(i,o)\in N\times O$ such that $\sum_{o':o'\succsim_i o}p(i,o')<1$ ($i$ would like to have more of $o$) and $\sum_{j\in N}p(j,o)<1$ ($o$ is not fully allocated).
Then, the object $o$ that is wasted at generalized random matching $p$ is not assigned to any agent in at least one of the generalized deterministic stable matchings $q^j$ in the convex combination. Thus, by the rural hospital theorem, $o$ is not assigned to any agent in any stable generalized deterministic matching in $\{q^1,\ldots, q^m\}$. Since $\sum_{o':o'\succsim_i o}p(i,o')<1$, it follows that in at least one of the stable generalized deterministic matchings $q^k$, $\sum_{o':o'\succsim_i o}q^k(i,o')=0$ and $\sum_{j\in N}q^k(j,o)=0$. Hence, $q^k$ is wasteful; a contradiction.
\end{proof}
		
A statement along the lines of Proposition~\ref{prop:ruralhospital-nonwastefulness} is not true anymore when preferences and priorities can be weak, as the following example demonstrates.

\begin{example}[\textbf{A wasteful and individually rational generalized random matching that satisfies inequalities (\ref{weakgeneralLE4W})}]\label{example2:wastefulIR+inequalities}
\normalfont Consider the following general instance $I=(N,O,\succsim)$ with weak preferences and weak priorities: $N=\{1,2,3\}$, $O=\{x,y,z\}$,   	
\begin{center}
\begin{tabular}{lccc}
$\succsim_1$:&$[x\ y]$&$\emptyset$& $z$\\
$\succsim_2$:&$[y\ z]$&$\emptyset$&$x$\\
$\succsim_3$:&$[x\ z]$&$\emptyset$&$y$
\end{tabular}
\quad\quad
\begin{tabular}{lccc}
$\succsim_x$:&$[1\ 3]$&$\emptyset$&$2$\\
$\succsim_y$:&$[1\ 2]$&$\emptyset$&$3$\\
$\succsim_z$:&$[2\ 3]$&$\emptyset$&$1$.
\end{tabular}	
\end{center}
Then, the generalized random matching
\begin{center}
\begin{tabular}{c}
$p=\begin{pmatrix}
\nicefrac{1}{3}&\nicefrac{1}{3}&0\\
0&\nicefrac{1}{3}&\nicefrac{1}{3}\\
\nicefrac{1}{3}&0&\nicefrac{1}{3}
\end{pmatrix}$
\end{tabular}										
\end{center}
is wasteful, individually rational, and satisfies inequalities (\ref{weakgeneralLE4W}).\hfill~$\diamond$
\end{example}

We have the following equivalence for fractional weak stability for generalized random matchings and their associated random matchings.

\begin{proposition}\label{Prop:weakfractionalstableimpliesfractionalstable} \label{Prop:fractionalstableimpliesfractionalweakstable}
The generalized random matching $p$ is fractionally weakly stable if and only if the associated random matching $p'$ is fractionally stable and respects non-wastefulness.
\end{proposition}

\begin{proof}[\textbf{Proof}]
Let $p$ be a generalized random matching and $p'$ its associated random matching.\medskip

\noindent\textbf{Part 1:} Let $p$ be a fractionally weakly stable generalized random matching. Thus, $p$ is non-wasteful, individually rational, and satisfies inequalities (\ref{weakgeneralLE4W}). Then, $p'$ respects non-wastefulness and individual rationality. Suppose, by contradiction, that $p'$ is not fractionally weakly stable. Then, for some pair $(a,b)\in N'\times O'$,
\begin{equation*}
\sum_{a':a'\prec'_{b} a}p'(a',b)>\sum_{b':b'\succsim'_a b;b'\neq b}p'(a,b').
\end{equation*}	
In particular,
\begin{equation*}
\sum_{a':a'\prec'_{b} a}p'(a',b)>0.
\end{equation*}	
Furthermore, recall that $\sum_{b':b'\succsim'_{a} b}p'(a,b')<1$ and hence,
\begin{equation*}
\sum_{b':b'\prec'_{a} b}p'(a,b')>0.
\end{equation*}	

\noindent \emph{Case~1.} Suppose that $b=o_j\in O$. Recall that
\begin{equation*}
\succsim'_{o_j}=\succsim_{o_j}\!\!\left(\{k\in N\midd k\succ_{o_j}\emptyset\}\right),\ d_j,\ lex\left(D\setminus \{d_j\}\right),\ \succsim_{o_j}\!\!\left(\{k\in N\midd \emptyset\succ_{o_j}k\}\right).
\end{equation*}
By the definition of $\succsim'$ and $p'$ and individual rationality (of $p$), for all $d_k\in D\setminus\{d_j\}$, $p'(d_k,o_j)=0$ and for all $l\in N$ such that $l\prec_{o_j}\emptyset$, $p'(l,o_j)=0$. Thus, if $a\precsim'_{o_j}d_j$, then  $\sum_{a':a'\prec'_{o_j} a}p'(a',o_j)=0$; a contradiction. Hence, $a\succ'_{o_j}d_j$ and $a=i\in N$ is an acceptable agent. By a symmetric argument, starting with $a=i\in N$ and
\begin{equation*}
\succsim'_{i}=\succsim_{i}\!\!\left(\{o\in O\midd o\succ_{i}\emptyset\}\right),\ \phi_i,\ lex\left(\Phi\setminus \{\phi_i\}\right),\ \succsim_{i}\!\!\left(\{o\in O\midd \emptyset\succ_{i}o\}\right),
\end{equation*}
we obtain $b\succ'_i\phi_i$ and that $b=o_j\in O$ is an acceptable object.

Then, by the definition of $\succsim'$ and $p'$ (recall that $p'(d_j,o_j)=p(\emptyset,o_j)$),
\begin{equation*}
i=a\succ'_{o_j}d_j\mbox{ implies }\sum_{a':a'\prec'_{o_j} i}p'(a',o_j)= \sum_{k:k\prec_{o_j} i}p(k,o_j)+p(\emptyset,o_j)
\end{equation*}	
and
\begin{equation*}
o_j=b\succ'_i\phi_i\mbox{ implies }\sum_{b':b'\succsim'_{i} o_j; b'\neq o_j}p'(i,b')= \sum_{o':o'\succsim_{i} o_j; o'\neq o_j}p(i,o').
\end{equation*}	
Hence, inequality $\sum_{a':a'\prec'_{b} a}p'(a',b)>\sum_{b':b'\succsim'_a b;b'\neq b}p'(a,b')$ for $a=i\in N$ and $b=o_j\in O$ can be rewritten as
\begin{equation*}
\sum_{k:k\prec_{o_j} i}p(k,o_j) +p(\emptyset,o_j)
> \sum_{o':o'\succsim_{i} o_j; o'\neq o_j}p(i,o'),\end{equation*}	
which contradicts that $p$ was fractionally weakly stable.

Since in Case~1 we have shown that $b\in O$ implies $a\in N$ and vice versa, the only remaining case to discuss is $(a,b)\in D\times\Phi$.\medskip

\noindent \emph{Case~2.} Suppose that $a=d_j\in D$ and $b=\phi_i\in \Phi$. Recall that
\begin{equation*}
\succsim'_{d_j}= o_j,\ lex\left(O\setminus \{o_j\}\right),\ \succsim'_{d_j}\!(\Phi).
\end{equation*}
By the definition of $\succsim'_{d_j}$ and $p'$, for all $l,m\in N$, $\phi_l\succsim'_{d_j}\phi_m$ if and only if $l\succsim_{o_j}m$ and $p'(d_j,o_j)=p(\emptyset,o_j)$. Then, we have
\begin{equation*}
\sum_{b':b'\succsim'_{d_j}\phi_i; b'\neq\phi_i}p'(d_j,b') =\sum_{k:\phi_k\succsim'_{d_j}\phi_i; k\neq i}p'(d_j,\phi_k) +p'(d_j,o_j) = \sum_{k:k\succsim_{o_j}i; k\neq i}p(k,o_j)+ p(\emptyset,o_j).
\end{equation*}
Next, recall that
\begin{equation*}
\succsim'_{\phi_i}= i,\ lex\left(N\setminus \{i\}\right),\ \succsim'_{\phi_i}\!(D).
\end{equation*}
By the definition of $\succsim'_{\phi_i}$, for all $x,y\in O$, $d_x\succsim'_{\phi_i}d_y$ if and only if $x\succsim_{i}y$.
Then, by the definition of $\succsim'_{\phi_i}$ and $p'$, we have
\begin{equation*}
\sum_{a':a'\prec'_{\phi_i} d_j}p'(a' ,\phi_i)=\sum_{d_{l}:d_l\prec'_{\phi_i} d_j}p'(d_l,\phi_i)=\sum_{o_{l}:o_l\prec_{i} o_j}p(i,o_l).
\end{equation*}
Hence, inequality $\sum_{a':a'\prec'_{b} a}p'(a',b)>\sum_{b':b'\succsim'_a b;b'\neq b}p'(a,b')$ for $a=d_j\in D$ and $b=\phi_i\in \Phi$ can be rewritten as
\begin{equation*}
\sum_{o_{l}:o_l\prec_{i} o_j}p(i,o_l)
> \sum_{k:k\succsim_{o_j}i; k\neq i}p(k,o_j)+ p(\emptyset,o_j).
\end{equation*}	
This implies   $\sum_{o_{l}:o_l\prec_{i} o_j}p(i,o_l)
> 0$ and individual rationality implies that agent $i$ finds object $o_j$ acceptable. Similarly it follows that $\sum_{k:k\succsim_{o_j}i; k\neq i}p(k,o_j)+ p(\emptyset,o_j)<1$, therefore $\sum_{k:k\precsim_{o_j} i}p(k,o_j)> 0$, and by individual rationality, object $o_j$ finds agent $i$ acceptable. Hence, $(i,o_j)\in N\times O$ is an acceptable pair.
Furthermore, recall that $\sum_{b':b'\succsim'_{a} b}p'(a,b')<1$ and hence,  $\sum_{o_{l}:o_l\succsim_{i} o_j}p(i,o_l)
< 1$. Thus, by non-wastefulness, $p(\emptyset,o_j)=0$. Therefore, for the acceptable pair $(i,o_j)\in N\times O$,
\begin{equation*}
\sum_{o_{l}:o_l\prec_{i} o_j}p(i,o_l)
> \sum_{k:k\succsim_{o_j}i; k\neq i}p(k,o_j)
\end{equation*}
and therefore also
\begin{equation*}
\sum_{o_{l}:o_l\prec_{i} o_j}p(i,o_l)+p(i,\emptyset)
> \sum_{k:k\succsim_{o_j}i; k\neq i}p(k,o_j);
\end{equation*}
contradicting that $p$ was fractionally weakly stable.\medskip

\noindent \textbf{Part 2:} Let $p'$ be a fractionally stable random matching that respects non-wastefulness. Thus, $p$ is non-wasteful. We first show that $p$ is individually rational. Consider an unacceptable pair $(i,o_j)\in N\times O$. Assume that object $o_j$ finds agent $i$ unacceptable, i.e., $\emptyset\succ_{o_j} i$. Now consider the pair $(d_j,o_j)\in N'\times O'$. Fractional stability of $p'$ requires
\begin{equation*}
  \sum_{b':b'\succsim'_{d_j}o_j;b'\neq o_j}p'(d_j,b')
\geq\sum_{a': a'\prec'_{o_j} d_j}p'(a',o_j).
\end{equation*}
Since object $o_j$ is the best object for $d_j$ at $\succsim'_{d_j}$, it follows that $\sum_{b':b'\succsim'_{i}d_j;b'\neq d_j}p'(d_j,b')
=0$. Hence, $\sum_{a':a'\prec'_{o_j} d_j}p'(a',o_j)=0$ and for each $a'\prec'_{o_j} d_j$, $p'(a',o_j)=0$. Next, $a'\prec'_{o_j} d_j$ if and only if $a'\in D\setminus\{d_j\}$ or [$a'\in N$ and $\emptyset\succ_{o_j} a'$]. Thus, by the definition of $p'$, for each $a'\in N$ such that $\emptyset\succ_{o_j} a'$, $p(a',o_j)=p'(a',o_j)=0$.
Symmetrically, starting from agent $i$ finding agent $o_j$ unacceptable, i.e., $\emptyset\succ_{o_j} i$, we obtain that for each $b'\in O$ such that $\emptyset\succ_i b'$, $p(i,b')=p'(i,b')=0$. Hence, the generalized random matching $p$ is individually rational.\medskip

Next suppose, by contradiction, that $p$ violates one of the inequalities (\ref{weakgeneralLE4W}). Then, for some acceptable pair $(i,o_j)\in N\times O$,
\begin{equation*}
\sum_{k:k\prec_{o_j} i}p(k,o_j)+p(\emptyset,o_j)>\sum_{o':o'\succsim_i o_j;o'\neq o_j}p(i,o').
\end{equation*}	
Recall that
\begin{equation*}
\succsim'_{o_j}=\succsim_{o_j}\!\!\left(\{k\in N\midd k\succ_{o_j}\emptyset\}\right),\ d_j,\ lex\left(D\setminus \{d_j\}\right),\ \succsim_{o_j}\!\!\left(\{k\in N\midd \emptyset\succ_{o_j}k\}\right)
\end{equation*}
and
\begin{equation*}
\succsim'_{i}=\succsim_{i}\!\!\left(\{o\in O\midd o\succ_{i}\emptyset\}\right),\ \phi_i,\ lex\left(\Phi\setminus \{\phi_i\}\right),\ \succsim_{i}\!\!\left(\{o\in O\midd \emptyset\succ_{i}o\}\right).
\end{equation*}
Then, by the definition of $\succsim'$ and $p'$ (recall that $p(\emptyset,o_j)=p'(d_j,o_j)$),
\begin{equation*}
i\succ_{o_j}\emptyset\mbox{ implies }\sum_{k:k\prec_{o_j} i}p(k,o_j)+p(\emptyset,o_j)=\sum_{a':a'\prec'_{o_j} i}p'(a',o_j)
\end{equation*}	
and
\begin{equation*}
o_j\succ_{i}\emptyset\mbox{ implies }\sum_{o':o'\succsim_i o_j;o'\neq o_j}p(i,o')=\sum_{b':b'\succsim'_{i} o_j; b'\neq o_j}p'(b',i).
\end{equation*}
Hence, inequality $\sum_{k:k\prec_{o_j} i}p(k,o_j)+p(\emptyset,o_j)>\sum_{o':o'\succsim_i o_j;o'\neq o_j}p(i,o')$  can be rewritten as
\begin{equation*}
\sum_{a':a'\prec'_{o_j} i}p'(a',o_j)
> \sum_{b':b'\succsim'_{i} o_j; b'\neq o_j}p'(b',i),
\end{equation*}	
which contradicts that $p'$ was fractionally stable.
\end{proof}

The following example demonstrates why we had to impose that $p'$ respects non-wastefulness in Proposition~\ref{Prop:fractionalstableimpliesfractionalweakstable}.

\begin{example}[\textbf{A wasteful and fractionally weakly stable associated random matching $\bm{p'}$}]\label{examle:wastefulfractionallystablep'}
\normalfont We consider the general instance $I=(N,O,\succsim)$ with weak preferences and weak priorities  (the brackets indicate indifferences) that we already have discussed in Example~\ref{example2:wastefulIR+inequalities}: $N=\{1,2,3\}$, $O=\{x,y,z\}$,
\begin{center}
\begin{tabular}{lccc}
$\succsim_1$:&$[x\ y]$&$\emptyset$& $z$\\
$\succsim_2$:&$[y\ z]$&$\emptyset$&$x$\\
$\succsim_3$:&$[x\ z]$&$\emptyset$&$y$
\end{tabular}
\quad\quad
\begin{tabular}{lccc}
$\succsim_x$:&$[1\ 3]$&$2$&$\emptyset$\\
$\succsim_y$:&$[1\ 2]$&$3$&$\emptyset$\\
$\succsim_z$:&$[2\ 3]$&$1$&$\emptyset$.
\end{tabular}	
\end{center}
Then, the generalized random matching
\begin{center}
\begin{tabular}{c}
$p=\begin{pmatrix}
\nicefrac{1}{3}&\nicefrac{1}{3}&0\\
0&\nicefrac{1}{3}&\nicefrac{1}{3}\\
\nicefrac{1}{3}&0&\nicefrac{1}{3}
\end{pmatrix}$
\end{tabular}										
\end{center}
is wasteful, individually rational, and satisfies inequalities (\ref{weakgeneralLE4W}) in the definition of fractional weak stability.

The associated instance $(N',O',\succsim')$ is such that $N'=\{1,2,3,d_x,d_y,d_z\}$, $O'=\{x,y,z,\phi_1,\phi_2,\phi_3\}$ with preferences and priorities (the brackets indicate indifferences): 	
\begin{center}
\begin{tabular}{lccccc}
$\succsim'_1$:&$[x\ y]$&$\phi_1$&$\phi_2$&$\phi_3$& $z$\\
$\succsim'_2$:&$[y\ z]$&$\phi_2$&$\phi_1$&$\phi_3$&$x$\\
$\succsim'_3$:&$[x\ z]$&$\phi_3$&$\phi_1$&$\phi_2$&$y$
\end{tabular}
\quad\quad
\begin{tabular}{lccccc}
$\succsim'_x$:&$[1\ 3]$&$2$&$d_x$&$d_y$&$d_z$\\
$\succsim'_y$:&$[1\ 2]$&$3$&$d_y$&$d_x$&$d_z$\\
$\succsim'_z$:&$[2\ 3]$&$1$&$d_z$&$d_x$&$d_y$
\end{tabular}\medskip

\begin{tabular}{lccccc}
$\succsim'_{\phi_1}$:&1&2&3&$[d_x\ d_y]$& $d_z$\\
$\succsim'_{\phi_2}$:&2&1&3&$[d_y\ d_z]$&$d_x$\\
$\succsim'_{\phi_3}$:&3&1&2&$[d_x\ d_z]$&$d_y$
\end{tabular}
\quad\quad
\begin{tabular}{lccccc}
$\succsim'_{d_x}$:&$x$&$y$&$z$&$[\phi_1\ \phi_3]$&$\phi_2$\\
$\succsim'_{d_y}$:&$y$&$x$&$z$&$[\phi_1\ \phi_2]$&$\phi_3$\\
$\succsim'_{d_z}$:&$z$&$x$&$y$&$[\phi_2\ \phi_3]$&$\phi_1$
\end{tabular}	
\end{center}
The associated random matching equals
\begin{center}
\begin{tabular}{c}
$p'=$\begin{blockarray}{ccccccccccc}
    &&\matindex{$x$} & \matindex{$y$} & \matindex{$z$}&&\matindex{$\phi_1$}&\matindex{$\phi_2$}&\matindex{$\phi_3$}\\
    \begin{block}{c(cccccccccc)}
   \matindex{$1$}& &$\nicefrac{1}{3}$   &$\nicefrac{1}{3}$  & $0$&$|$&$\nicefrac{1}{3}$&$0$&$0$&\\
   \matindex{$2$}& & $0$  & $\nicefrac{1}{3}$  &$\nicefrac{1}{3}$&$|$&$0$&$\nicefrac{1}{3}$&$0$\\
   \matindex{$3$} & & $\nicefrac{1}{3}$& $0$  & $\nicefrac{1}{3}$&$|$&$0$&$0$&$\nicefrac{1}{3}$\\
       &&---  & ---&--- &---&---&---&---&\\
   \matindex{$d_x$} && $\nicefrac{1}{3}$& $0$ & $0$&$|$&$\nicefrac{1}{3}$&$0$&$\nicefrac{1}{3}$&\\
       \matindex{$d_y$}& & $0$ &$\nicefrac{1}{3}$ & $0$&$|$&$\nicefrac{1}{3}$&$\nicefrac{1}{3}$&$0$&\\
         \matindex{$d_z$}& &$0$ & $0$ & $\nicefrac{1}{3}$&$|$&$0$&$\nicefrac{1}{3}$&$\nicefrac{1}{3}$&\\
    \end{block}
  \end{blockarray}
\end{tabular}										
\end{center}
and does not respect non-wastefulness.

One can now check for each $(a,b)\in N'\times O'$ that the fractional stability inequalities (\ref{LE4W}) are satisfied and hence $p'$ is fractionally stable. However, since $p$ is wasteful, it is not fractionally weakly stable.\hfill~$\diamond$
\end{example}

Next, in order to define claimwise weak stability for generalized random matchings, the notion of a claim can be adjusted as follows: using consumption process language, as long as agent $i$ consumes objects that are different from and not worse than $o$ he does not envy lower priority agent $j$ to consume fractions of $o$ and he does not mind fractions of $o$ to be unassigned, however, once the unassigned amounts of $o$ plus the amount lower priority agent $j$ has consumed reach agent $i$'s weak upper contour set at $o$ (not including $o$), agent $i$ either envies agent $j$ or complains about wastefulness (unless agent $i$ can fill his remaining probability quota with object $o$). An agent $i\in N$ has a \emph{claim} against an agent $j\in N$, if there exists an object $o\in O$ such that $(i,o)$ is an acceptable pair, $i\succ_o j$, and
\begin{equation}\label{weaklyLE7New}
p(j,o)+p(\emptyset,o)>\sum_{o':o'\succsim_i o;o'\neq o}p(i,o').
\end{equation}
Inequality (\ref{weaklyLE7New}) implies $\sum_{o':o'\succsim_i o}p(i,o')<1$, i.e., agent $i$ receives some fraction of an object in his strict lower contour set at $o$ or $i$ is not fully matched (if not, this would imply that $\sum_{o':o'\succsim_i o;o'\neq o}p(i,o')+p(i,o)=1$ and hence, $p(j,o)+p(\emptyset,o)+p(i,o)>1$; a contradiction). Thus, agent $i$ would want to consume more of object $o$.\medskip

A generalized random matching is \emph{claimwise weakly stable} if it is individually rational, non-wasteful, and does not admit any claim.

\begin{definition}[\textbf{Claimwise weak stability for generalized random matchings}]
\normalfont	A generalized random matching $p$ is \emph{\textbf{claimwise weakly stable}} if
$p$ is individually rational, non-wasteful, and
for each acceptable pair $(i,o)\in N\times O$ and each $j\in N$ such that $i\succ_oj$,
\begin{equation}\label{weakLE8general}
\sum_{o':o'\succsim_i o;o'\neq o}p(i,o') \geq  p(j,o)+ p(\emptyset,o).
\end{equation}
\label{def:weakgeneralclaimwisestable}
\end{definition}

With the next proposition and example we show that only one direction of the transformation between the base model and the most general model preserves claimwise weak stability, while the other does not. The intuitive reason that an equivalence result as in the case of fractional weak stability (Proposition~\ref{Prop:weakfractionalstableimpliesfractionalstable}) does not hold for claimwise weak stability (Proposition~\ref{Prop:weakclaimwisestableiffclaimwisestable}) is as follows: fractional weak stability is a \textit{symmetric} notion in that a violation that involves agent $i$ who would like more of object $o$ when facing lower priority agents is equivalent to a violation that involves object $o$ wanting more of agent $i$ when facing lower preferred objects while, in contrast, a claim is \textit{one-sidedly} defined by an agent $i$ wanting more of object $o$ when facing one lower priority agent without any implications for object $o$ wanting more of agent $i$ when facing one lower preferred object.

\begin{proposition}\label{Prop:weakclaimwisestableiffclaimwisestable}
The generalized random matching $p$ is claimwise weakly stable if  the associated random matching $p'$ is claimwise stable and respects non-wastefulness and individual rationality.
\end{proposition}

\begin{proof}[\textbf{Proof}]
Let $p$ be a generalized random matching and $p'$ its associated random matching. Let $p'$ be claimwise stable and respect non-wastefulness and individual rationality. Thus, $p$ is non-wasteful and individual rational. Suppose, by contradiction, that $p$ violates one of the inequalities (\ref{weakLE8general}). Then, for some acceptable pair $(i,o_j)\in N\times O$  and some agent $k\in N$ such that $k\prec_{o_j} i$,
\begin{equation*}
p(k,o_j)+p(\emptyset,o_j)>\sum_{o':o'\succsim_i o_j;o'\neq o_j}p(i,o').
\end{equation*}
Furthermore, $\sum_{o':o'\succsim_i o}p(i,o')<1$ and hence, by non-wastefulness, $p(\emptyset,o_j)=0$. Recall that
\begin{equation*}
\succsim'_{o_j}=\succsim_{o_j}\!\!\left(\{k\in N\midd k\succ_{o_j}\emptyset\}\right),\ d_j,\ lex\left(D\setminus \{d_j\}\right),\ \succsim_{o_j}\!\!\left(\{k\in N\midd \emptyset\succ_{o_j}k\}\right)
\end{equation*}
and
\begin{equation*}
\succsim'_{i}=\succsim_{i}\!\!\left(\{o\in O\midd o\succ_{i}\emptyset\}\right),\ \phi_i,\ lex\left(\Phi\setminus \{\phi_i\}\right),\ \succsim_{i}\!\!\left(\{o\in O\midd \emptyset\succ_{i}o\}\right).
\end{equation*}
Then, by the definition of $\succsim'$ and $p'$ (recall that $p(\emptyset,o_j)=p'(d_j,o_j)=0$),
\begin{equation*}
k\prec'_{o_j} i\mbox{ and }p(k,o_j)+p(\emptyset,o_j)=p'(k,o_j)
\end{equation*}	
and
\begin{equation*}
o_j\succ_{i}\emptyset\mbox{ implies }\sum_{o':o'\succsim_i o_j;o'\neq o_j}p(i,o')=\sum_{b':b'\succsim'_{i} o_j; b'\neq o_j}p'(b',i).
\end{equation*}
Hence, inequality $p(k,o_j)+p(\emptyset,o_j)>\sum_{o':o'\succsim_i o_j;o'\neq o_j}p(i,o')$  can be rewritten as
\begin{equation*}
p'(k,o_j)
> \sum_{b':b'\succsim'_{i} o_j; b'\neq o_j}p'(b',i),
\end{equation*}	
which contradicts that $p'$ was claimwise stable.
\end{proof}

\begin{example}[\textbf{A non-wasteful, individually rational, and claimwise weakly stable generalized random matching $\bm{p}$ but $p'$ is not claimwise stable}]\label{examle:weakclaimiseone-sided}\normalfont We reconsider the example used in the proof of Proposition~\ref{prop:claimwise-notto-expost}. Let $N=\{1,2,3\}$ and $O=\{x,y,z\}$. Consider the following preferences and priorities: 	
\begin{center}
\begin{tabular}{lcccc}
$\succ_1$:&$x$&$z$&$y$&$\emptyset$\\
$\succ_2$:&$y$&$x$&$z$&$\emptyset$\\
$\succ_3$:&$z$&$x$&$y$&$\emptyset$
\end{tabular}
\quad\quad
\begin{tabular}{lcccc}
$\succ_x$:&$2$&$3$&$1$&$\emptyset$\\
$\succ_y$:&$1$&$3$&$2$&$\emptyset$\\
$\succ_z$:&$2$&$1$&$3$&$\emptyset$
\end{tabular}	
\end{center}				
Let $p$ be the uniform random matching. Thus,
\begin{center}
\begin{tabular}{c}
$p=\begin{pmatrix}
\nicefrac{1}{3}&\nicefrac{1}{3}&\nicefrac{1}{3}\\
\nicefrac{1}{3}&\nicefrac{1}{3}&\nicefrac{1}{3}\\
\nicefrac{1}{3}&\nicefrac{1}{3}&\nicefrac{1}{3}
\end{pmatrix}.$
\end{tabular}
\end{center}	
Random matching $p$ is claimwise stable (see proof of Proposition~\ref{prop:claimwise-notto-expost}), individually rational, and non-wasteful.\medskip

The associated instance is $I'=(N',O',\succsim')$ where $N'=\{1,2,3,d_x,d_y,d_z\}$, $O'=\{x,y,z,\phi_1,\phi_2,\phi_3\}$, with preferences and priorities:
\begin{center}
\begin{tabular}{lcccccc}
$\succsim'_1$:&$x$& $z$&$y$&$\phi_1$&$\phi_2$&$\phi_3$\\
$\succsim'_2$:&$y$& $x$&$z$&$\phi_2$&$\phi_1$&$\phi_3$\\
$\succsim'_3$:&$z$& $x$&$y$&$\phi_3$&$\phi_1$&$\phi_2$\\
\end{tabular}
\quad\quad
\begin{tabular}{lcccccc}
$\succsim'_x$:&$2$&$3$&$1$&$d_x$&$d_y$&$d_z$\\
$\succsim'_y$:&$1$&$3$&$2$&$d_y$&$d_x$&$d_z$\\
$\succsim'_z$:&$2$&$1$&$3$&$d_z$&$d_x$&$d_y$
\end{tabular}\medskip

\begin{tabular}{lcccccc}
$\succsim'_{\phi_1}$:&1&2&3&$d_x$& $d_z$&$d_y$\\
$\succsim'_{\phi_2}$:&2&1&3&$d_y$& $d_x$&$d_z$\\
$\succsim'_{\phi_3}$:&3&1&2&$d_z$& $d_x$&$d_y$\\
\end{tabular}
\quad\quad
\begin{tabular}{lccccccc}
$\succsim'_{d_x}$:&$x$&$y$&$z$&$\phi_2$&$\phi_3$&$\phi_1$\\
$\succsim'_{d_y}$:&$y$&$x$&$z$&$\phi_1$&$\phi_3$&$\phi_2$\\
$\succsim'_{d_z}$:&$z$&$x$&$y$&$\phi_2$&$\phi_1$&$\phi_3$
\end{tabular}	
\end{center}
The associated random matching is
\begin{center}
\begin{tabular}{c}
$p'=$\begin{blockarray}{ccccccccccc}
    &&\matindex{$x$} & \matindex{$y$} & \matindex{$z$}&&\matindex{$\phi_1$}&\matindex{$\phi_2$}&\matindex{$\phi_3$}\\
    \begin{block}{c(cccccccccc)}
   \matindex{$1$}& &$\nicefrac{1}{3}$   &$\nicefrac{1}{3}$  & $\nicefrac{1}{3}$&$|$&$0$&$0$&$0$&\\
   \matindex{$2$}& & $\nicefrac{1}{3}$  & $\nicefrac{1}{3}$  &$\nicefrac{1}{3}$&$|$&$0$&$0$&$0$\\
   \matindex{$3$} & & $\nicefrac{1}{3}$& $\nicefrac{1}{3}$  & $\nicefrac{1}{3}$&$|$&$0$&$0$&$0$\\
       &&---  & ---&--- &---&---&---&---&\\
   \matindex{$d_x$} && $0$& $0$ & $0$&$|$&$\nicefrac{1}{3}$&$\nicefrac{1}{3}$&$\nicefrac{1}{3}$&\\
       \matindex{$d_y$}& & $0$ &$0$ & $0$&$|$&$\nicefrac{1}{3}$&$\nicefrac{1}{3}$&$\nicefrac{1}{3}$&\\
         \matindex{$d_z$}& &$0$ & $0$ & $0$&$|$&$\nicefrac{1}{3}$&$\nicefrac{1}{3}$&$\nicefrac{1}{3}$&\\
    \end{block}
  \end{blockarray}\ .
\end{tabular}
\end{center}
By definition, $p'$ respects non-wasteful and individually rational with respect to $p$. However, $p'$ is not claimwise stable: agent $d_x$ has a justified claim against $d_z$ for $\phi_2$ because $d_x\succ_{\phi_2} d_z$, $p'(d_z,\phi_2)=1/3$ and $\sum_{o':o'\succ_{d_x}  \phi_2}p'(i,o')=0$.\hfill~$\diamond$
\end{example}

Example~\ref{examle:wastefulfractionallystablep'} can also be used to demonstrate why we had to impose that $p'$ respects non-wastefulness in Proposition~\ref{Prop:weakclaimwisestableiffclaimwisestable}: the associated random matching $p'$ in the example is also claimwise weakly stable and does not respect non-wastefulness. Hence, the underlying generalized random matching $p$ is wasteful and hence not weakly claimwise stable. The following example demonstrates why we had to impose that $p'$ respects individual rationality in Proposition~\ref{Prop:weakclaimwisestableiffclaimwisestable}.

\begin{example}[\textbf{An individually irrational and claimwise weakly stable  associated random matching $\bm{p'}$}]\label{examle:associated-claimwise-notIR}\normalfont
Let $N=\{1,2\}$ and $O=\{x,y\}$. Consider the following preferences and priorities: 	
\begin{center}
\begin{tabular}{lcccc}
$\succ_1$:&$x$&$\emptyset$&$y$\\
$\succ_2$:&$[x,y]$&$\emptyset$\\
\end{tabular}
\quad\quad
\begin{tabular}{lcccc}
$\succ_x$:&$[1,2]$&$\emptyset$\\
$\succ_y$:&$[1,2]$&$\emptyset$\\
\end{tabular}
\end{center}
Let $p$ be the uniform random matching. Thus,
\begin{center}
\begin{tabular}{c}
$p=\begin{pmatrix}
\nicefrac{1}{2}&\nicefrac{1}{2}\\
\nicefrac{1}{2}&\nicefrac{1}{2}
\end{pmatrix}.$
\end{tabular}
\end{center}
Random matching $p$ is individually irrational, non-wasteful and satisfies inequalities (\ref{weakLE8general}) in the definition of claimwise weak stability.\medskip

The associated instance is $I'=(N',O',\succsim')$ where $N'=\{1,2,d_x,d_y\}$, $O'=\{x,y,\phi_1,\phi_2\}$, with preferences and priorities:
\begin{center}
\begin{tabular}{lcccccc}
$\succsim'_1$:&$x$&$\phi_1$&$\phi_2$&$y$&\\
$\succsim'_2$:&$[x,y]$&$\phi_1$&$\phi_2$
\end{tabular}
\quad\quad
\begin{tabular}{lcccccc}
$\succsim'_x$:&$[1,2]$&$d_x$&$d_y$\\
$\succsim'_y$:&$[1,2]$&$d_y$&$d_x$\\
\end{tabular}
\medskip\\
\begin{tabular}{lcccccc}
$\succsim'_{\phi_1}$:&1&2&$d_x$&$d_y$\\
$\succsim'_{\phi_2}$:&2&1&$[d_x,d_y]$
\end{tabular}
\quad\quad
\begin{tabular}{lccccccc}
$\succsim'_{d_x}$:&$x$&$y$&$[\phi_1,\phi_2]$\\
$\succsim'_{d_y}$:&$y$&$x$&$[\phi_1,\phi_2]$\\
\end{tabular}
\end{center}
The associated random matching equals
\begin{center}
\begin{tabular}{c}
$p'=$\begin{blockarray}{ccccccccc}
    &&\matindex{$x$} & \matindex{$y$} &&\matindex{$\phi_1$}&\matindex{$\phi_2$}\\
    \begin{block}{c(cccccccc)}
   \matindex{$1$}& &$\nicefrac{1}{2}$   &$\nicefrac{1}{2}$  & $|$&$0$&$0$\\
   \matindex{$2$}& & $\nicefrac{1}{2}$  & $\nicefrac{1}{2}$  &$|$&$0$&$0$\\
       &&---  & ---&--- &---&---\\
   \matindex{$d_x$} && $0$& $0$ & $|$&$\nicefrac{1}{2}$&$\nicefrac{1}{2}$&\\
       \matindex{$d_y$}& & $0$ &$0$ & $|$&$\nicefrac{1}{2}$&$\nicefrac{1}{2}$&\\
    \end{block}
  \end{blockarray}
\end{tabular}
\end{center}
and does not respect individual rationality. We argue that $p'$ is claimwise stable.
Agent~1 gets $1/2$ of $x$ and thus does not have a justified claim for $\phi_1$ or $\phi_2$ against $d_x$ or $d_y$.
Agent~2 gets a best possible outcome and thus has no justified claim.
Agents~$d_x$ or $d_y$ cannot have a justified claim against agent 1 or 2 because the latter have higher priority.
Finally, agents~$d_x$ and $d_y$ have no justified claim against each other.\hfill~$\diamond$
\end{example}

It follows easily that if we restrict attention to generalized deterministic matchings, then all the stability concepts for generalized random matchings coincide with standard weak stability (Definition~\ref{def:weakgeneralstable}). The proof of Proposition~\ref{generalproposition1weak} follows the same arguments as the proof of our previous Propositions~\ref{proposition1} and \ref{proposition1weak} and we therefore omit it.

\begin{proposition}\label{generalproposition1weak}
For generalized deterministic matchings, all the stability concepts for generalized random matchings with weak preferences and weak priorities coincide with weak stability for deterministic matchings.
\end{proposition}

Our previous results (Figure~\ref{fig:weakpart-relations} together with Propositions~\ref{Prop:weakexpostStableiffexpostStable} -- \ref{Prop:weakclaimwisestableiffclaimwisestable}) now imply the following taxonomy of the stability concepts for generalized random matchings and their associated random matching in Figure~\ref{fig:weakassociated-random}.

\begin{figure}[H]%[!htbp]
\begin{framed}
\textbf{Section~\ref{subsection:weakdifferentnumbers} results:} 	For any generalized  random matching $p$ and its associated random matching $p'$, we have\bigskip
\begin{center}
\scalebox{0.9}{
\begin{tikzpicture}
\tikzstyle{pfeil}=[->,>=angle 60, shorten >=1pt,draw]
\tikzstyle{onlytext}=[]
    \node[onlytext] (p) at (0,0) {\begin{tabular}{cl}
 & no ex-ante envy,\\     $p$ &and non-wastefulness,\\
 & individual rationality\\
 &(Defs.~\ref{def:weakgeneralnoexantejustifiedenvy}, \ref{def:weakgeneralnonwastefulness}, and \ref{def:weakgeneralindividualrationality})
\end{tabular}
};

    \node[onlytext] (p'exante) at (5,-4) {$p'$ ex-ante weakly stable (Def.~\ref{def:weakex-antestability})};
	\node[onlytext] (p'rexpost) at (5,-6) {\begin{tabular}{cl} &$p'$ robust ex-post weakly stable (Def.~\ref{def:weakrobustexpoststability})\\
	&and non-wasteful
	\end{tabular}
	};
	\node[onlytext] (p'expost) at (5,-8) {\begin{tabular}{cl}&$p'$ ex-post weakly stable (Def.~\ref{def:weakexpoststability})\\
	&and non-wasteful
	\end{tabular}
	};
	\node[onlytext] (p'fractional) at (5,-10) {\begin{tabular}{cl}&$p'$ fractionally  weakly stable (Def.~\ref{def:weakfractionalstability})\\
	&and non-wasteful
	\end{tabular}
	};
	
 \node[onlytext] (Ap'claimwise) at (1.2,-11.8) {};
	
		\node[onlytext] (p'claimwise) at (5,-12) {\begin{tabular}{cl} &$p'$ claimwise weakly stable (Def.~\ref{def:weaklyclaimwisestable}),\\
	&non-wasteful and individually rational
	\end{tabular}
		};

    \draw[pfeil, thick, blue, bend right] (p'exante) to (p'rexpost);
	\draw[pfeil, thick, red, bend right] (p'rexpost) to (p'exante);
	\draw[pfeil, thick, blue, bend right] (p'rexpost) to (p'expost);
	\draw[pfeil, thick, red, bend right] (p'expost) to (p'rexpost);
	\draw[pfeil, thick, red, bend right] (p'fractional) to (p'expost);
	\draw[pfeil, thick, blue, bend right] (p'expost) to (p'fractional);
	\draw[pfeil, thick, blue, bend right] (p'fractional) to (p'claimwise);
	\draw[pfeil, thick, red, bend right] (p'claimwise) to (p'fractional);
	
	\draw[thick, red] (5.25,-11.2) to (5.65,-10.9);
	\draw[thick, red] (5.25,-9.2) to (5.65,-8.9);
	\draw[thick, red] (5.25,-7.2) to (5.65,-6.9);
	\draw[thick, red] (5.25,-5.2) to (5.65,-4.7);
	
    \node[onlytext] (pexante) at (-5,-4) {$p$ ex-ante weakly stable (Def.~\ref{def:weakgeneralexantestable})};
	\node[onlytext] (prexpost) at (-5,-6) {$p$ robust ex-post weakly stable (Def.~\ref{def:weakgeneralrobustexpoststability})};
	\node[onlytext] (pexpost) at (-5,-8) {$p$ ex-post weakly stable (Def.~\ref{def:weakgeneralexpostStability})};
	\node[onlytext] (pfractional) at (-5,-10) {$p$ fractional  weakly stable (Def.~\ref{def:weakgeneralfractionalstable})};
	\node[onlytext] (pclaimwise) at (-5,-12) {$p$ claimwise stable (Def.~\ref{def:weakgeneralclaimwisestable})};
	
\node[onlytext] (Apclaimwise) at (-2.4,-11.8) {};

	\draw[pfeil, thick, blue, bend right] (pexante) to (prexpost);
	\draw[pfeil, thick, red, bend right] (prexpost) to (pexante);
	\draw[pfeil, thick, blue, bend right] (prexpost) to (pexpost);
	\draw[pfeil, thick, red, bend right] (pexpost) to (prexpost);
	\draw[pfeil, thick, red, bend right] (pfractional) to (pexpost);
	\draw[pfeil, thick, blue, bend right] (pexpost) to (pfractional);
	\draw[pfeil, thick, blue, bend right] (pfractional) to (pclaimwise);
	\draw[pfeil, thick, red, bend right] (pclaimwise) to (pfractional);

	\draw[thick, red] (-4.8,-11.2) to (-4.4,-10.8);
	\draw[thick, red] (-4.8,-9.2) to (-4.4,-8.8);
	\draw[thick, red] (-4.8,-7.2) to (-4.4,-6.8);
	\draw[thick, red] (-4.8,-5.2) to (-4.4,-4.8);

\draw[pfeil, thick, blue] (p) to (p'exante);
\draw[pfeil, thick, blue] (p'exante) to (p);
 \node[onlytext] (pp') at (3.7,-2) {Prop.~\ref{Prop:weakex-anteNJE-IR-NWiffex-anteStable}};

\draw[pfeil, thick, blue] (pexante) to (p'exante);
\draw[pfeil, thick, blue] (p'exante) to (pexante);
 \node[onlytext] (pp') at (0,-4.5) {Prop.~\ref{Prop:weakex-anteStableiffex-anteStable}};

 \draw[pfeil, thick, blue] (prexpost) to (p'rexpost);
 \draw[pfeil, thick, blue] (p'rexpost) to (prexpost);
   \node[onlytext] (pp') at (0,-6.5) {Prop.~\ref{Prop:weakrexpostStableiffrexpostStable}};

 \draw[pfeil, thick, blue] (pexpost) to (p'expost);
 \draw[pfeil, thick, blue] (p'expost) to (pexpost);
  
     \node[onlytext] (pp') at (0,-8.5)
    {Prop.~\ref{Prop:weakexpostStableiffexpostStable}};

 \draw[pfeil, thick, blue] (pfractional) to (p'fractional);
 \draw[pfeil, thick, blue] (p'fractional) to (pfractional);
  \node[onlytext] (pp') at (0,-10.5) {Prop.~\ref{Prop:weakfractionalstableimpliesfractionalstable}};

\node[onlytext] (pp') at (-0.3,-11.2) {Ex.~\ref{examle:weakclaimiseone-sided}};
\draw[thick, red] (-0.4,-11.6) to (-0.8,-12.0);
   \draw[pfeil, thick, red] (Apclaimwise) to (Ap'claimwise);
  \draw[pfeil, thick, blue] (p'claimwise) to (pclaimwise);
    \node[onlytext] (pp') at (0,-12.5) {Prop.~\ref{Prop:weakclaimwisestableiffclaimwisestable}};

\draw[pfeil, thick, blue] (p) to (pexante);
\draw[pfeil, thick, blue] (pexante) to (p);
 \node[onlytext] (pp-stable) at (-3.5,-2) {Prop.~\ref{Prop:weakex-anteStableWiffNJE-IR-NW}};

\end{tikzpicture}}
\end{center}	
\caption{Relations between stability concepts for generalized random matchings with weak preferences and weak priorities and equivalences of stability concepts for associated random matchings.}
\label{fig:weakassociated-random}
\end{framed}
\end{figure}

\newpage

\section{Conclusion}\label{section:conclusion}

We presented a taxonomy of stability concepts (ex-ante; robust ex-post, ex-post, fractional, and claimwise) for the most well-studied but restricted setting in which (1) preferences are strict, (2) priorities are strict, (3) there is an equal number of agents and objects, (4) all objects and agents are acceptable to each other. The formalization lead to a clear picture of the hierarchy of stability concepts. We then extended these concepts to the most general model that has none of the restrictions (1) -- (4).
We formalized the stability concepts  with the appropriate additional requirements of non-wastefulness and/or individual rationality when necessary to preserve the hierarchy we established in the base model.
We found that it was an extremely subtle task to identify when additionally requiring individual rationality or non-wastefulness is redundant or when it is critical to preserve the logical relations and characterizations that were identified in the base model.
We also took these factors into account when obtaining our characterization results for preserving stability concepts when transforming the most general model to the base model.
Throughout the paper, we complement our results with minimal examples where converse statements do not hold or when a certain characterization cannot be extended.
We are hopeful that the groundwork in this paper will provide the base for further market design and axiomatic work on probabilistic matching under priorities.

\appendix

\section{Appendix: Weak and strong stochastic dominance stability \citep{Manj13a} re-examined}\label{appendixstochstab}

In this section, we point out connections with weak and strong stochastic dominance (sd) stable matchings as studied by \citet{Manj13a} for the base model as introduced in Section~\ref{section:model} (with an equal number of agents and objects and strict preferences / priorities). Note that our model involves ordinal preferences of agents over objects and ordinal priorities of objects over agents. These preferences / priorities can be extended to preferences / priorities over random allocations via the first order stochastic dominance relation.

\newpage

{\begin{definition}[\textbf{First order stochastic dominance}]\normalfont
Given two random matchings $p$ and $q$ and an agent $i\in N$ with preference $o_1\succ_i o_2 \succ_i \ldots \succ_i o_n$ over $O=\{o_1,\ldots,o_n\}$, we say that agent $i$ \textbf{\emph{$\bm{\sd}$-prefers}} match $p(i)$ to match $q(i)$, denoted by $p(i) \succsim_i^{\sd} q(i)$, if and only if,
 % \item[\textrm{(i)}]$o_1\succ_i o_2 \succ_i \ldots \succ_i o_n$ and
\[\begin{array}{rcl}
                       p(i,o_1) & \geq & q(i,o_1) \\
                       p(i,o_1) + p(i,o_2) & \geq & q(i,o_1) + q(i,o_2) \\
                       p(i,o_1) + p(i,o_2) + p(i,o_3) & \geq & q(i,o_1) + q(i,o_2) + q(i,o_3)\\
                       &\vdots&
                     \end{array}\]
If $p(i) \succsim_i^{\sd} q(i)$ and $p(i)\neq q(i)$, then $p(i) \succ_i^{\sd} q(i)$.\medskip

Given two random matchings $p$ and $q$ and an object $o\in O$ with priorities $i_1\succ_o i_2 \succ_o\ldots \succ_o i_n$ over $N=\{i_1,\ldots, i_n\}$, we say that object $o$ \textbf{\emph{$\bm{\sd}$-prioritizes}} match $p(o)$ to match $q(o)$, denoted by $p(o) \succsim_o^{\sd} q(o)$, if and only if
\[\begin{array}{rcl}
                       p(i_1,o) & \geq & q(i_1,o) \\
                       p(i_1,o) + p(i_2,o) & \geq & q(i_1,o) + q(i_2,o) \\
                       p(i_1,o) + p(i_2,o) + p(i_3,o) & \geq & q(i_1,o) + q(i_2,o) + q(i_2,o)\\
                       &\vdots&
                     \end{array}\]
	          If $p(o) \succsim_o^{\sd} q(o)$ and $p(o)\neq q(o)$, then $p(o) \succ_o^{\sd} q(o)$.
\end{definition}

The definitions of Manjunath's weak and strong stochastic dominance stability are based on the following two pairwise blocking notions.

\begin{definition}[\textbf{Weak and strong (pairwise) sd-blocking; \citeauthor{Manj13a}, \citeyear{Manj13a}}]\normalfont A random matching $p$ is \textbf{\emph{weakly sd-blocked}} by pair $(i,o)\in N\times O$ if there exists a corresponding deterministic matching $q\neq p$ such that $q(i,o)=1$ and $$\mbox{neither }p(i)\succ_i^{sd}q(i)\mbox{ nor }p(o)\succ_o^{sd}q(o).$$
A random matching $p$ is \textbf{\emph{strongly sd-blocked}} by pair $(i,o)\in N\times O$ if there exists a corresponding deterministic matching $q\neq p$ such that $q(i,o)=1$ and $$q(i)\succ_i^{sd}p(i)\mbox{ and }q(o)\succ_o^{sd}p(o).$$
\end{definition}

\begin{definition}[\textbf{Weak and strong sd-stability; \citeauthor{Manj13a}, \citeyear{Manj13a}}]\normalfont A random matching $p$ is \textbf{\emph{weakly sd-stable}} if there exists no pair $(i,o)\in N\times O$ that strongly sd-blocks $p$.

A random matching $p$ is \textbf{\emph{strongly sd-stable}} if there exists no pair $(i,o)\in N\times O$ that weakly sd-blocks $p$.
\end{definition}
	
\begin{proposition}
A random matching is strongly sd-stable if and only if it is ex-ante stable. 	
\end{proposition}

\begin{proof}[\textbf{Proof}]Suppose random matching $p$ has ex-ante envy. Then, there exist $i,j\in N$ and $o,o'\in O$ such that $p(i,o')>0$, $p(j,o)>0$, $o\succ_i o'$, and $i\succ_o j$. Consider a corresponding deterministic matching $q\neq p$ such that $q(i,o)=1$. Thus, neither $p(i)\succ_i^{sd}q(i)$ nor $p(o)\succ_o^{sd}q(o)$. Hence, $p$ is weakly sd-blocked by pair $(i,o)$ and not strongly sd-stable.

Suppose random matching $p$ is not strongly sd-stable. Then, there exists a pair $(i,o)\in N\times O$ that weakly sd-blocks $p$, i.e., there exists a corresponding deterministic matching $q\neq p$ such that $q(i,o)=1$ and neither $p(i)\succ_i^{sd}q(i)$ nor $p(o)\succ_o^{sd}q(o)$. Note $\sum_{o':o'\succsim_i o}p(i,o')=1$ would imply $p(i)\succ_i^{sd}q(i)$ and  $\sum_{j:j\succsim_o i}p(j,o)=1$ would imply $p(o)\succ_o^{sd}q(o)$. Thus, $\sum_{o':o'\succsim_i o}p(i,o')<1$ and $\sum_{j:j\succsim_o i}p(j,o)<1$. Then, there exist $j\in N$ and $o'\in O$ such that $p(i,o')>0$, $p(j,o)>0$, $o\succ_i o'$, and $i\succ_o j$. Hence, $p$ is not ex-ante stable.	\end{proof}
	
\begin{proposition}
If a random matching $p$ is claimwise weakly stable, then it is weakly sd-stable.
\end{proposition}

\begin{proof}[\textbf{Proof}]Suppose random matching $p$ is not weakly sd-stable. Then, there exists a pair $(i,o)\in N\times O$ that strongly sd-blocks $p$, i.e., there exists a corresponding deterministic matching $q\neq p$ such that $q(i,o)=1$ and $q(i)\succ_i^{sd}p(i)$ and $q(o)\succ_o^{sd}p(o)$. Note that $q(o)\succ_o^{sd}p(o)$ implies $\sum_{j:k\succ_o i}p(k,o)=0$  and $p(i,o)<1$. Hence, there exists an agent $j\in N$ such that $i\succ_o j$ and $p(j,o)>0$. Furthermore, $q(i)\succ_i^{sd}p(i)$ implies $\sum_{o':o'\succ_io}p(i,o')=0$. Thus, $p(j,o)>\sum_{o':o'\succ_io}p(i,o')$ and agent $i$ has a claim against agent $j$ and $p$ is not claimwise stable.
\end{proof}

\begin{proposition}\label{prop:sdstability-notto-claimwise}
Weak sd-stability does not imply claimwise stability.
\end{proposition}	
\begin{proof}[\textbf{Proof}]
Let $N=\{1,2,3\}$ and $O=\{x,y,z\}$. Consider the following preferences and priorities: 	
\begin{center}
\begin{tabular}{lccc}
$\succ_1$:&$x$&$z$&$y$\\
$\succ_2$:&$y$&$x$&$z$\\
$\succ_3$:&$z$&$x$&$y$
\end{tabular}
\quad\quad
\begin{tabular}{lccc}
$\succ_x$:&$2$&$3$&$1$\\
$\succ_y$:&$1$&$3$&$2$\\
$\succ_z$:&$2$&$1$&$3$
\end{tabular}	
\end{center}
Consider the random matching
\begin{center}
\begin{tabular}{c}
$p=\begin{pmatrix}
\nicefrac{1}{4}&\nicefrac{1}{2}&\nicefrac{1}{4}\\
\nicefrac{1}{4}&\nicefrac{1}{4}&\nicefrac{1}{2}\\
\nicefrac{1}{2}&\nicefrac{1}{4}&\nicefrac{1}{4}
\end{pmatrix}.$
\end{tabular}										
\end{center}	
First, note that agent 2 wants more of object $x$, $2\succ_x 3$, and $p(3,x)=\frac{1}{2}>\frac{1}{4}=\sum_{o:o\succ_2x}p(2,o)$. Hence, agent 2 has a claim against agent 3 and $p$ is not claimwise stable.

Second, we show that random matching $p$ is weakly sd-stable by checking that for no pair $(i,o)\in N\times O$ with corresponding deterministic matching $q\neq p$ such that $q(i,o)=1$, $q(i)\succ_i^{sd}p(i)$ and $q(o)\succ_o^{sd}p(o).$
\begin{itemize}
\item For an agent $i\in N$ and his most preferred object $o\in O$, $q(o)\not\succ_o^{sd}p(o)$ because all other agents have higher priority for that object.
\item For an agent $i\in N$ and his second or third preferred object, $q(i)\not\succ_i^{sd}p(i)$ because agent $i$ receives his best object with positive probability.
\end{itemize}				
\end{proof}
%\begin{singlespace}
\bibliographystyle{chicago}
% \bibliography{abb,taxonomy}
%\end{singlespace}

\end{document}